%% file: pods19.tex
\documentclass[sigconf,10pt]{acmart}

\usepackage{balance}

\def\BibTeX{{\rm B\kern-.05em{\sc i\kern-.025em b}\kern-.08emT\kern-.1667em\lower.7ex\hbox{E}\kern-.125emX}}
    
\copyrightyear{2019}
\acmYear{2019}
\setcopyright{rightsretained}

\acmConference[PODS'19]{38th ACM SIGMOD-SIGACT-SIGAI Symposium on Principles of Database Systems}{June 30-July 5, 2019}{Amsterdam, Netherlands}
\acmBooktitle{38th ACM SIGMOD-SIGACT-SIGAI Symposium on Principles of Database Systems (PODS'19), June 30-July 5, 2019, Amsterdam, Netherlands}
\acmPrice{}
\acmDOI{10.1145/3294052.3319695}
\acmISBN{978-1-4503-6227-6/19/06}

\acmSubmissionID{pods74}

\usepackage{nicefrac}       
\usepackage{microtype}      
\usepackage[]{algorithm}
\usepackage{algorithmic}

\input{macros}

\begin{document}
\title{What Storage Access Privacy is Achievable with Small Overhead?}

\author{Sarvar Patel}
\affiliation{
  \institution{Google LLC}
}
\email{sarvar@google.com}

\author{Giuseppe Persiano}
\affiliation{
  \institution{Universit\`a di Salerno}
  \institution{Google LLC}
}
\email{giuper@gmail.com}

\author{Kevin Yeo}
\affiliation{
  \institution{Google LLC}
}
\email{kwlyeo@google.com}

%

\begin{abstract}\input{abstract}\end{abstract}

\maketitle

\input{intro}

\input{results}

\input{related}

\input{defs}

\input{lower}
\input{lower_ir}
\input{lower_ram}

\input{lower_disc}

\input{strawman}
\input{ir}

\input{ram}
\input{ram_proof}

\input{kvs}
\input{kvs_map_ram}
\input{kvs_map}

\input{concl}

\bibliographystyle{abbrv}
\bibliography{biblio}

\appendix
\input{tools}

\input{dpir-proof}

\input{ir_low_mult-app}

\input{ram_leftover_proofs}
\input{generalization}

\input{kvs_leftover_proofs}

\input{ir_up_single-app}

\input{ram_up-app}

\end{document}

%% file: macros.tex
\newcommand{\bzero}{{\mathbf{0}}}
\newcommand{\prev}{\mathtt{pr}}
\newcommand{\nxt}{\mathtt{nx}}
\DeclareMathOperator{\dist}{\text{\sf D}}
\newcommand{\ch}{\calC}
\newcommand{\trans}{\calT}
\newcommand{\poly}{\mathsf{poly}}
\newcommand{\pirC}{{\sf{DP\text{-}}\ir}}
\newcommand{\adv}{\calA}
\DeclareMathOperator{\E}{\mathbb{E}}
\newcommand{\range}{{\mathsf{Range}}}
\newcommand{\calA}{\mathcal{A}}
\newcommand{\calQ}{\mathcal{Q}}

\newcommand{\key}{\mathtt{key}}

\newcommand{\calT}{\mathcal{T}}

\newcommand{\calD}{\mathcal{D}}
\newcommand{\calC}{{\mathcal{C}}}
\DeclareMathOperator{\negl}{\mathtt{negl}}

\newcommand{\calS}{\mathcal{S}}
\newcommand{\textOp}{\text{\sf op}}
\newcommand{\tOp}{\textOp}
\newcommand{\opRead}{\text{\sf read}}
\newcommand{\opR}{\opRead}
\newcommand{\opWrite}{\text{\sf write}}
\newcommand{\opW}{\opWrite}

\newcommand{\ir}{\mathsf{IR}}
\newcommand{\ram}{\mathsf{RAM}}
\newcommand{\kvs}{\mathsf{KVS}}
\newcommand{\dpir}{{\mathsf{DP\text{-}\ir}}}
\newcommand{\edpir}{{\epsilon\mathsf{\text{-}\dpir}}}
\newcommand{\eddpir}{{(\epsilon, \delta)\mathsf{\text{-}\dpir}}}
\newcommand{\dpram}{{\mathsf{DP\text{-}\ram}}}
\newcommand{\oramC}{\dpram}
\newcommand{\edpram}{{\epsilon\mathsf{\text{-}\dpram}}}
\newcommand{\eddpram}{{(\epsilon,\delta)\mathsf{\text{-}\dpram}}}
\newcommand{\dpkvs}{{\mathsf{DP\text{-}KVS}}}
\newcommand{\edpkvs}{{\epsilon\mathsf{\text{-}DP\text{-}KVS}}}
\newcommand{\eddpkvs}{{(\epsilon,\delta)\mathsf{\text{-}\dpkvs}}}

\DeclareMathOperator{\Enc}{\text{\sf Enc}}
\DeclareMathOperator{\Dec}{\text{\sf Dec}}
\newcommand{\Init}{\text{\sf Setup}}

\newcommand{\server}{{\sf{server}}}
\newcommand{\client}{{\sf{client}}}
\DeclareMathOperator{\bStash}{\text{\sf bStash}}
\DeclareMathOperator{\Query}{\text{\sf Query}}
\newcommand{\textIndex}{\ensuremath{i}}

%% file: abstract.tex

Oblivious RAM (ORAM) and private information retrieval (PIR)
are classic cryptographic primitives used to
hide the access pattern to data whose storage has been outsourced to
an untrusted server.
Unfortunately, both primitives require considerable overhead
compared to plaintext access. For large-scale storage infrastructure
with highly frequent access requests, the degradation in response time
and the exorbitant increase
in resource costs incurred by
either ORAM or PIR prevent their usage.
In an ideal scenario, a privacy-preserving storage protocols with
small overhead would be implemented for these heavily trafficked storage systems
to avoid negatively impacting either performance and/or costs.
In this work, we study the problem of the best {\em storage access privacy} that is achievable
with only {\em small overhead} over plaintext access.

To answer this question,
we consider {\em differential privacy access} which is a generalization
of the {\em oblivious access} security notion that are considered by ORAM and PIR.
Quite surprisingly, we present strong evidence that constant overhead storage schemes
may only be achieved with privacy budgets of $\epsilon = \Omega(\log n)$.
We present asymptotically optimal constructions for differentially private variants of
both ORAM and PIR with privacy budgets $\epsilon = \Theta(\log n)$ with
only $O(1)$ overhead. In addition, we consider a more complex storage primitive
called key-value storage in which data is indexed by keys from a large universe
(as opposed to consecutive integers in ORAM and PIR).
We present a differentially private key-value storage scheme with
$\epsilon = \Theta(\log n)$ and $O(\log\log n)$ overhead. This construction
uses a new oblivious, two-choice hashing scheme that may be of independent interest.



%% file: intro.tex
\section{Introduction}

Privacy-preserving storage primitives consider the outsourcing
of the storage of a database to an untrusted server with the ability for clients
to retrieve database records while maintaining
the privacy of the retrievals.
Even though encryption can be used to hide the content of the database,
patterns of how the encrypted data is accessed are revealed.
The leakage of access patterns has been
shown to compromise privacy in many important practical
settings~\cite{IKK12, CGP15}.
Privacy-preserving storage primitives that guarantee retrieval privacy
have been used as a critical component
in many systems such as advertisement~\cite{GLM16},
discovery of identities~\cite{BDG15}
and publish-subscribe~\cite{CSP16}.
Therefore, a very important question involves the construction
of privacy-preserving storage schemes guaranteeing retrieval privacy while
ensuring that record retrieval can be performed efficiently.

A common way to formulate the privacy of retrievals is {\em obliviousness}.
Obliviousness guarantees that for any two fixed sequences of record retrievals
of the same length,
any adversary that views all accesses to stored data
cannot determine which of the two sequences induced
the resulting access pattern to stored data.
Obliviousness has been considered with both statistical security,
providing protection from adversaries with unbounded computational resources,
and computational security,
where the adversary is assumed to be probabilistic polynomial time (PPT).
Oblivious RAM (ORAM) and private information retrieval (PIR) are two
oblivious storage primitives that have been the objective of extensive
research~\cite{Goldreich87,CGK95,GO96,KO97,CMS99,GR05,GM11,KLO12,SVS13,PPR18}.
For ORAM, it has been shown that $\Omega(\log n)$ overhead is
necessary~\cite{LN18}. On the other hand, the best constructions for PIR require
at least $\Omega(n)$ server computation over the entire outsourced database.

For integral storage infrastructure that handle many access requests per second,
the increased response time and server costs caused by the usage of
ORAM and PIR prevent their implementation in these important systems.
In a perfect world, there would exist a storage scheme with strong privacy with only
small overhead that could be implemented for these frequently accessed
storage infrastructures without negatively affecting the performance and/or expenses.
In our work, we address the natural question of the best privacy that may be achieved
by any storage scheme with only small overhead compared to plaintext access.
To our knowledge,
our work is the first that considers the question of maximizing privacy
for a specific efficiency goal. Previous works consider minimizing the
efficiency for primitives with a specific privacy notion.

To consider this problem, we use another formulation for storage access privacy
through {\em differential privacy}~\cite{DMN06, Dwork11, DR14}.
Typically, differential privacy is used in the context of
{\em privacy-preserving data analysis} where global properties of
the entire database are disclosed while
maintaining the privacy of individual database records.
In particular, differential privacy guarantees that any fixed disclosure
is just as likely, usually within a multiplicative factor, regardless
of whether an individual belonged to the sample population or not.
Our work focuses on the notion of {\em differentially private access}
which attempts to
maintain privacy for individual record retrievals,
but may reveal information about the entire sequence of retrievals.
Roughly speaking,
differential privacy guarantees that
any sequence of accesses to stored data caused by
the execution of a sequence of record retrievals will be
just as likely, except for a multiplicative factor,
caused by another sequence attained by replacing
a single retrieval for a database record with a single retrieval for any other
database record. Differentially private access
is a generalization of oblivious access as oblivious access
is achieved by fixing
the differential privacy parameters to be $\epsilon = 0$ and
$\delta = \negl(n)$.

With this definition,
readers might now ask the following two questions.
When does differentially private
access make sense as a security notion? Or, might there exist a stronger
security notion that is achievable with small overhead?

To answer the first question, we revisit the scenario of
privacy-preserving data analysis on outsourced databases.
In these scenarios, there is no sense in using storage schemes
providing obliviousness to hide entire record retrieval sequences when
differentially private disclosures only maintain privacy for single records.
For example, suppose we wish to disclose a differentially private
model trained over a sample from the database.
Obliviousness would unnecessarily hide the identity of the entire retrieved sample at a high
cost yet the differential privacy would guarantee the privacy about individuals in the sample.
Therefore, differentially private access is the privacy notion that is
complementary to differential privacy disclosures on outsourced databases.
In general, differentially private access guarantees privacy for individual retrievals.
For the second question, recent work by Persiano and Yeo~\cite{PY18} show
that $\Omega(\log n)$ overhead is necessary for differentially private RAMs
with parameters of $\epsilon = O(1)$ and $\delta \le 1/3$.
Therefore, the task of finding good parameters for
differential privacy that provide both meaningful
privacy as well as while permitting small overhead seems non-trivial.


We study differentially private variants of both PIR and ORAM.
PIR enables clients to obliviously retrieve
database records outsourced to a server.
In PIR, both the client and the server are stateless
which means that no information may be maintained between multiple
record retrievals beyond the server storing the database.
Since both the client and the server are stateless, PIR requires
the server to perform an operation on each database record:
if a record is not involved in computing the server's reply, then
it cannot be the record retrieved by the client.
The majority of PIR constructions use homomorphic
encryption~\cite{ABF16, ACL17} or other expensive encryptions with
useful properties~\cite{KO97, CMS99, GR05}.
Some recent works consider PIR with stateful clients~\cite{PPY18b}
or super-linear server storage~\cite{CHR17,BIP17,HOW18}.
All single-server PIR schemes consider computational security.
PIR has also been studied in the multiple, non-colluding server setting
where constructions only require servers to perform computation
sublinear in the number of database records and provide
statistical security~\cite{CGK95,BIM00}.

On the other hand, ORAM allows both the client and server to be stateful
and maintain information between multiple queries and allows
the client to perform both record retrievals and overwrites.
Additionally, ORAM allows a setup phase where
the untrusted server receives an encrypted version of the database to store
and the client is given a secret key.
As a result of state, the server is no longer required to perform
computation on every database record.
With this efficiency, the majority of ORAM schemes consider only simple
upload and download operations and forego the use of expensive encryption
schemes of PIR that require the server to perform untrusted computation,
although some previous works consider ORAM with homomorphic
encryption~\cite{DDF16}.
For a database with $n$ records, it has been shown
that ORAM requires overhead of $\Omega(\log n)$ records~\cite{GO96,LN18}
and that $O(\log n \cdot \log\log n)$ communication suffices~\cite{PPR18}.
We also examine an extension to ORAM, which we denote as {\em oblivious key-value storage}
(previously also denoted as oblivious storage), where database records
are uniquely identified by keys from a large universe and clients
might also attempt to retrieve a non-existent key~\cite{GMO12}.

In our work, we study the differentially private variants of these three primitives, which
we denote as {\em differentially private information retrieval} ($\dpir$),
{\em differentially private RAM} ($\dpram$) and
{\em differentially private key value storage} ($\dpkvs$).
We focus on the question the best privacy that can be achieved
by each of these primitives that only require small overhead over plaintext access.

%% file: results.tex
\noindent
{\bf Our Contributions.}
We present both negative and positive results for
both $\dpir$ and $\dpram$ as well as other natural
variants of these primitives.

Our lower bounds for
$\dpir$ and $\dpram$ apply for a wide range of privacy budgets
when the database is stored
in a natural encoding.
In particular, we consider the balls and bins model of
storage (previously considered in~\cite{GO96,BN16,CCM18}) where
each database record is considered as an opaque ball along with a key containing
important metadata.
While not covering all database encodings, in our opinion, the
balls and bins models encompasses all natural database representations
that maximize practical utility and efficiency.
In particular, the contents of each record are assumed to be placed together,
which is typically done to maximize data locality.

For positive results, we show that a class of simple
constructions are optimal for large sets
of privacy budgets.
These schemes may be viewed as inserting noise into the sequence of
record retrievals and/or overwrites.
In particular, the real record retrievals and/or overwrites
are disguised within a set that also contain fake record retrievals and/or
overwrites.
These schemes suffice to construct the best privacy-preserving
storage protocols with very small overhead for $\dpir$, $\dpram$
and $\dpkvs$.
While our constructions are simple, we draw attention to the fact that
designing differentially private storage systems is delicate even with
weak security notions.
Some simple constructions are very appealing and, at first, seem to match our
lower bounds. However, many variants of simple constructions (including our constructions)
end up being completely insecure.
As an example, in Section~\ref{sec:insecure}, we consider a simple and tempting $\dpir$ construction and show that it only guarantees differential privacy with $\delta\rightarrow 1$,
i.e., no privacy at all.
Additionally, while our constructions are simple, the security
proofs end up being quite involved (especially for $\dpram$).
Furthermore, to handle the additional functionalities of $\dpkvs$,
a more complex algorithm using a novel, oblivious adaptation of
the two-choice hashing scheme~\cite{Mitzenmacher01} is required.

\paragraph{DP-IR Results.}
To our knowledge, previous
works on $\dpir$ consider only the multiple, non-colluding server
scenarios~\cite{TDG16}. Our work is the first to consider $\dpir$ in the single
server scenario.

We show that for any $\edpir$ that must always output the right answer
and any value of $\epsilon \ge 0$,
then the server must operate on all stored $n$
records.
This result is very strong and somewhat surprising as there
does not exist any weakening of privacy that will improve
the server computational costs.

\begin{theorem}[informal]
For any $\epsilon, \delta \ge 0$, any $\eddpir$ scheme in the balls and bins
model must operate on $(1-\delta)n$ records.
\end{theorem}

On the other hand, we show that this strong negative result may be circumvented
by considering $\dpir$ schemes
with non-zero error probabilities $0 < \alpha \le 1$.
Here, the error probability is over the internal randomness of the $\dpir$
scheme and does not depend on the queries and/or stored data.
For this case, we present the following weaker lower bound.
\begin{theorem}[informal]
For any $\epsilon, \delta \ge 0$, any $\eddpir$ scheme in the balls and bins model
with error probability $\alpha > 0$ must operate on
$$
\Omega\left( \frac{(1-\alpha - \delta) \cdot n}{e^\epsilon} \right)
$$
records.
\end{theorem}

As we are focusing on schemes with very small overhead, the above
theorem leads credulence that there might exist a $\edpir$ scheme
with some small, constant error probability $\alpha > 0$
that only performs $O(1)$ operations when $\epsilon = \Theta(\log n)$.
We show that there exists a simple construction with these properties.

\begin{theorem}[informal]
There exists an $\edpir$ with $\epsilon = \Theta(\log n)$ and
constant error probability $\alpha > 0$ that only operates on $O(1)$ records.
\end{theorem}

Due to our lower bounds, the above construction seems to be the best
privacy that can be achieved by any $\dpir$ with only constant overhead
compared to plaintext storage access.

We also consider natural extensions of our $\dpir$ with
multiple, non-colluding servers.
There exist several
multiple-server PIR schemes~\cite{CGK95,BIM00} in literature.
In our full paper, we present asymptotically tight lower bounds for the constructions
in~\cite{TDG16}.

\smallskip
\paragraph{DP-RAM Results.}
For $\dpram$, we once again start with describing our lower bound results.
Unlike $\dpir$, there is no separation between the best lower bound
for perfectly correct $\dpram$ and $\dpram$ with error probability $\alpha > 0$.
We now present our $\dpram$ lower bound which applies for all values
$0 \le \alpha \le 1$.

\begin{theorem}[informal]
For any $\epsilon \ge 0$, any $\edpram$ with error probability $\alpha \ge 0$
in the balls and bins model and a client that stores at most $c$ blocks must operate on
$$
\Omega\left(\log_c \left(\frac{(1-\alpha) \cdot n}{e^\epsilon}\right)\right)
$$
records.
\end{theorem}

The above theorem essentially states that there are two ways that
one can achieve very efficient $\edpram$ constructions:
either increase the amount of client storage ($c$) or increase
the privacy budget ($\epsilon$).
In most scenarios, it is impractical to suppose that the client
can store large portions of data. As our desire is to construct
schemes that should be easily usable in complex systems,
we minimize the requirements of our clients by assuming clients
have very small amounts of storage.
In~\cite{PY18}, an $\Omega(\log (n/c))$ lower bound is given for
constant $\epsilon$ and $\delta \le 1/3$. However, their lower bound
can be generalized to $\Omega(\log (n/c) / e^\epsilon)$ for any
$\epsilon \ge 0$.
Their lower bound has an exponentially worse dependence on $\epsilon$ compared
to our lower bound.
For example, the lower bound in~\cite{PY18} does not preclude the existence
of a $\edpram$ with $\epsilon = \Theta(\log\log n)$ and constant overhead.
On the other hand, our lower bound improves the privacy budget lower bound
showing that an $\edpram$ that operates on $O(1)$ records must have
$\epsilon = \Omega(\log n)$.
We show the existence of a constant overhead $\edpram$ with asymptotically optimal
$\epsilon = \Theta(\log n)$ privacy.

\begin{theorem}[informal]
There exists an $\edpram$ with $\epsilon = \Theta(\log n)$ that only operates
on $O(1)$ records.
\end{theorem}

Once again, this construction seems to be the best privacy
that can be achieved by any $\dpram$ scheme with only $O(1)$ overhead
over the baseline, unprotected storage access due to our lower bounds.
Our scheme improves on prevous $\dpram$ schemes in~\cite{WCM16}
which starts from Path ORAM~\cite{SVS13} and degrades security to improve efficiency.
For their scheme to achieve even client storage of $O(\sqrt{n})$,
their construction recursively stores position maps
which costs both logarithmic overhead and client-to-server roundtrips.
On the other hand, our $\dpram$ construction uses both $O(1)$ overhead and roundtrips
while achieving small client storage.


\paragraph{DP-KVS Results.}
Finally, we consider $\dpkvs$ which is an extension of $\dpram$. As a result,
all $\dpram$ lower bounds also apply to $\dpkvs$. Therefore, the best
construction that is achievable by $\dpkvs$ with $O(1)$ overhead
could be $\epsilon = \Theta(\log n)$.
Due to the difficulties of handling a larger universe of queries
and possibility that clients request on-existent keys, we present
a construction with slightly worse than constant overhead. In addition,
we can only achieve privacy in the approximate differential privacy
framework.

\begin{theorem}[informal]
There exists a $\edpkvs$ with $\epsilon = \Theta(\log n)$
that operates on $O(\log\log n)$ records.
\end{theorem}

While the scheme has non-constant overhead, the $O(\log\log n)$ overhead
is exponentially smaller than the best oblivious key-value storage
schemes based on ORAMs.

To construct our $\dpkvs$ scheme, we present
an improved, oblivious variant of the power of two choices
hashing scheme~\cite{Mitzenmacher01} that may be of separate,
independent interest.
Traditional power of two choices hashing guarantees that bins do not exceed
$O(\log\log n)$ items except with probability negligible in $n$.
Our desired oblivious variant must hide the sizes of bins.
One way to hide bin sizes is to pad all bins with dummy items up to a maximum.
This technique requires padding all $n$ bins to $O(\log\log n)$ items meaning
$O(n\log\log n)$ storage.
By using a tree-like structure to allow bins to share storage, we present
a new scheme using $O(n)$ storage.

%% file: related.tex
\noindent{\bf Related Work.}
$\dpram$ was considered previously in~\cite{WCM16} which present
a construction based on Path ORAM~\cite{SVS13}. However, their scheme
requires recursively stored position maps which requires $\Theta(\log n)$
client-to-server roundtrips to get client storage of even $O(\sqrt{n})$.
We present $\dpram$ schemes that only require $O(1)$ overhead with
small client storage.
A construction of $\dpir$ in the multiple, non-colluding server scenario
was considered in~\cite{TDG16}, which we show is optimal for
certain parameters.
Lower bounds for $\dpram$ have been considered in~\cite{PY18},
which are stronger and weaker in different dimensions. Their lower bound
has an exponentially worse dependence on the privacy budget, $\epsilon$.
On the other hand, their lower bound applies to general storage encodings
whereas our work only apply to the balls and bins model.
A variant of $\dpram$ that only maintains privacy
for database record insertions is considered in~\cite{CCM18}.
A privacy notion stronger than obliviousness
is considered in~\cite{KKN17} where accesses are
protected using obliviousness while the number of accesses is protected
using differential privacy.
Obliviousness has been considered for other problems such as
sorting~\cite{Batcher68,AKS83} and shuffling~\cite{Waksman68,OGT14,PPY18c}.
The problem of relaxing the security notion of obliviousness to differential
privacy has also been studied in the context of multi-party protocols.
A simple multi-party protocol for single-bit inputs is shown
in~\cite{NIPS2015_6004} that maximizes the accuracy for any privacy budget.
The multi-party computation problem for larger input sizes is studied
in~\cite{MG17}. Additionally, the problem of differentially private disclosure
of specific analysis such as subspace clustering~\cite{NIPS2015_5991},
deep learning~\cite{ACG16} and many others have also been studied
where the database is assumed to be stored on a trusted server.

%% file: defs.tex
\section{Definitions}
\label{sec:defs}
We suppose the database $D$ contains $n$ records denoted $B_1,\ldots,B_n$.
We will interchangeably use the terms records and blocks. We will refer to a
query as a single operation involving either a record retrieval or overwrite
and a query sequence as a list consisting of record retrievals and/or
overwrites.
We use $[n]$ to refer to the set $\{1, \ldots, n\}$.
We refer to $\calQ$ as the space of all possible queries and let
$Q_1,Q_2 \in \calQ^{l}$ be two query sequences of length $l$.
We define the Hamming distance between $Q_1$ and $Q_2$, which we denote
by $d(Q_1,Q_2)$, as the number of queries where $Q_1$ and $Q_2$ differ.

\subsection{Storage Primitives}
The {\em information retrieval} ($\ir$) primitive stores a
database $D$ of $n$ equal sized records where only record
retrievals are allowed.
The initialization of $\ir$ consists of the server receiving
$D = (B_1,\ldots,B_n)$ and simply processing and storing $D$.
The server is only allowed to store
the database and may not keep any other information between multiple queries.
Similarly, the client is stateless and may not use any storage between
multiple queries. A query to $\ir$ is described using an integer
$q \in [n]$,
which is interpreted as retrieving record $B_q$.
We denote two query sequences $Q_1, Q_2 \in [n]^{l}$ as adjacent
if their Hamming distance is exactly 1.
That is, $Q_1$ and $Q_2$ retrieve a different record
at exactly one query.

The {\em random access memory} ($\ram$) primitive will store a database
$D$ of $n$ equal sized records where both record retrievals and overwrites
are permitted. The initialization of $\ram$ consists of a setup phase
consisting of a protocol run between the client and the server.
The client will receive a private key, process $D$ using the
private key and send the processed database to the server to store.
Both the client and server may be stateful and maintain information
between multiple queries. A query to $\ram$ is a pair $q = (i, \tOp)$
where $i \in [n]$ refers to record $B_i$ and $\tOp \in \{\opR, \opW\}$
describes whether the query is a retrieval or overwrite.
Two query sequences $Q_1, Q_2 \in ([n] \times \{\opR, \opW\})^l$ are adjacent
if their Hamming distance is exactly 1. That is, at exactly one query,
$Q_1$ and $Q_2$ operate on a different record and/or perform a different
operation.

Finally, the {\em key-value storage} ($\kvs$) primitive
is an extension to $\ram$. Queries to $\kvs$ consist of a pair
$q = (k, \tOp)$ where $k \in U$ is the key and $U$ is universe of
all keys and $\tOp \in \{\opR, \opW\}$ refers to whether the query
is a retrieval or overwrite.
Unlike $\ram$, the universe of keys is very large and might be
exponentially larger than the number of operations that will
be performed.
Furthermore, a retrieval operation $q = (k, \opR)$ may request
a key $k$ that has never been previoulsy inserted into the storage
protocol. In this case, the $\kvs$ protocol should output $\perp$.
Identical to $\ram$, two query sequences $Q_1,Q_2 \in (U \times \{\opR, \opW\})^l$ are adjacent if there exists exactly one query that operates
on a different key and/or performs a different operation.

We note that $\ir$, $\ram$ and $\kvs$ are the most studied storage
primitives to provide oblivious access.
There are many other extensions to these primitives
that have also been
studied~\cite{GMO12, MBN15, BCR16, CLT16, CGL17, CHR17, BIP17, HOW18}.

\subsection{Differentially Private Access to Data}
Our privacy notion for a storage primitive $\calS$ storing a database
$D = (B_1,\ldots,B_n)$ with a query space
$\calQ$ will consider the random
variable of the {\em view} of the adversarial server for a query sequence
$Q\in\calQ^l$ of length $l$. With a slight abuse of notation,
we refer to the transcript, $\calS(Q)$,
as the random variable of the adversary's {view} on query sequence $Q$.
The transcript contains all movement of records performed by the server,
as well as all, possibly encrypted, records that have been uploaded and downloaded 
and the initial database.

\begin{definition}[Differentially Private Access]
Let $\calS$ be a storage primitive with query space $\calQ$.
$\calS$ provides $(\epsilon,\delta)$-{\em differentially private access}
if for all pairs of query sequences $Q_1,Q_2 \in \calQ^l$ of length $l$ that are adjacent,
that is $d(Q_1,Q_2)=1$,
and for any subset $S$ of the set of possible views of the adversary, then
the following holds:
\begin{align*}
\Pr[\calS(Q_1) \in S] \le e^\epsilon \cdot \Pr[\calS(Q_2) \in S] + \delta.
\end{align*}
\end{definition}

The $\epsilon$ parameter is referred to as the {\em privacy budget}.
When $\delta = 0$, the above definition is referred to as
{\em pure differential privacy} and the $\delta$ is typically
dropped from notation.
We denote $\ir$, $\ram$ and $\kvs$ primitives that provide
pure differentially private
access with privacy budget $\epsilon$ as $\edpir$, $\edpram$ and $\edpkvs$.
When $\delta > 0$, the definition is referred to as
{\em approximate differential privacy}, which is a
weakening of pure differential privacy.
We denote $\ir$, $\ram$ and $\kvs$ primitives that provide differentially
private access with parameters $\epsilon$ and $\delta$ as
$\eddpir$, $\eddpram$ and $\eddpkvs$.

%% file: lower.tex
\section{Lower Bounds}
\label{sec:low}
In this section, we present negative results about $\dpir$ and $\dpram$
in the single-server setting. As $\dpkvs$ is an extension of $\dpram$
with more functionality requirements, all $\dpram$ lower bounds apply
directly to $\dpkvs$.

\subsection{Balls and Bins Model}
Our lower bounds are presented in the balls and bins model of data manipulation.
The $n$ database records are treated as immutable, opaque balls.
Each ball is associated with a mutable key containing metadata about the record.
The formal definition involves a $\server_m$ with storage for $m$ balls
and a $\client_c$ with storage for $c$ balls.

\begin{definition}[Balls and Bins Model]
A $\client_c$ and $\server_m$ operate in the balls and bins model
if all client memory is initially empty and client-server
interactions are restricted to the following:
\begin{enumerate}
\item
Download ball from $\server_m$ to $\client_c$. For some
$i \in [m]$ and $j \in [c]$, store the block at address $i$ at $\server_m$
in address $j$ at $\client_c$.
\item
Upload ball from $\client_c$ to $\server_m$. For some $i \in [m]$ and
$j \in [c]$, store the block at address $j$ at $\client_c$ in address
$i$ at $\server_m$.
\end{enumerate}
\end{definition}

The above definition assumes a passive server acting only as storage.
As a result, only lower bounds on communication can be proven directly.
However, for algorithms with general computation by the server,
lower bounds on communication in the balls and bins model may be modified
to provide lower bounds on server computation. The modification simply views
the transcript as the balls that must be operated on by the server.

\paragraph{Discussion about storage model.}
The balls and bins model does not include all possible encodings of databases
that may be stored. In particular, we assume that database records are all
stored together and the contents do not emit special properties.
In general, related information should be stored in nearby memory locations
so that all required data can be found with the minimal number of cache misses.
Therefore, assuming database records are stored together captures real
world scenarios.
In addition, opaque balls rule out bothersome corner cases where
non-trivial lower bounds cannot hold such as when
records are dependent and record contents may be generated using superficial methods.

%% file: lower_ir.tex
\subsection{DP-IR}
\label{sec:low_ir}

We prove our lower bounds directly for approximately differential privacy
(in this case, $\delta \ge 0$). However, one can interpret the results for
pure differential privacy by setting $\delta = 0$.

To prove our lower bounds, we first define the notion of {\em transcripts}.
For a single-server
information retrieval protocol $\ir$ in the balls and bins model,
the transcript $\ir(i)$ is
the random variable denoting the set of blocks requested when retrieving $B_i$.
The query algorithm of an IR algorithm in the balls and bins model only issues
download commands to the server.

The main observation to prove the above theorem is that
the initialization phase of $\dpir$ is public which implies that the adversary
has knowledge of the identity of each record. By the restriction on
correctness, the queried record must always be retrieved.
Using the privacy requirement, it turns out all other
blocks must also be retrieved.
We summarize this in the following main lemma.
\begin{lemma}
\label{lemma:pir1}
Let $\ir$ be an ${\eddpir}$ scheme in the balls and bins model.
For any $a, b, c \in [n]$,
\begin{enumerate}
\item
$\Pr[B_c \in \ir(a)] \le e^\epsilon \Pr[B_c \in \ir(b)] + \delta.$
\item
$\Pr[B_c \notin \ir(a)] \le e^\epsilon \Pr[B_c \notin \ir(b)] + \delta.$
\end{enumerate}
\end{lemma}
\begin{proof}
Define $S \subseteq \range(\ir)$ to be all transcripts where
the block $B_c$ is downloaded. So, for any query $q$,
$$
\Pr[B_c \in \ir(q)] = \Pr[\ir(q) \in S].
$$
By $(\epsilon,\delta)$-DP, we know
$$
\Pr[\ir(a) \in S] \le e^\epsilon \Pr[\ir(b) \in S] + \delta
$$
completing the first point. The second point follows identically if we choose
$S \subseteq \range(\ir)$ to be all transcripts where block $B_c$ is not
downloaded.
\end{proof}

Using the above lemma,
we present very strong negative results for $\dpir$ that are errorless.
That is, they always return the correct record.
We show the following result:
\begin{theorem}
\label{thm:low_ir_no_error}
If $\ir$ is an $\eddpir$ scheme in the balls and bins model
then $\ir$ performs at least $(1-\delta)n$ operations in expectation.
\end{theorem}
\begin{proof}
By Lemma~\ref{lemma:pir1}, we know that
$$
\Pr[B_j \notin \ir(i)] \le e^\epsilon \Pr[B_j \notin \ir(j)] + \delta = \delta
$$
since $\ir$ is errorless and $\Pr[B_j \notin \ir(j)] = 0$.
Therefore, we know that $\Pr[B_j \in \ir(i)] \ge 1 - \delta$ and
$$
\E[|\ir(i)|] = \sum\limits_{j \in [n]} \Pr[B_j \in \ir(i)] \ge (1-\delta)n
$$
completing the proof.
\end{proof}

For algorithms with general server computation, we can the interpret the
above results as a lower bound on server computation.
This result is extremely strong since $n$ server operations must be executed
even when increasing the privacy budget. Therefore, the relaxation
to $\edpir$ does not result in any gain compared to PIR.
For $\eddpir$, one could increase $\delta$ to decrease costs.
However, typical privacy requires $\delta = \negl(n)$ resulting
in almost no gain.

To circumvent this result, we move to the case where $\dpir$ has a non-zero
error rate $0 < \alpha \le 1$.
That is, $\dpir$ will only retrieve the desired record with probability
$1 - \alpha$ depending only on the internal randomness of $\dpir$.
Our hope is that a very small $\alpha$ may significantly improve
efficiency of $\dpir$ and bypass this negative result. We show:
\begin{theorem}
\label{thm:low_ir}
If $\ir$ is an $\eddpir$ scheme in the balls and bins model
with error probability $\alpha > 0$, then $\ir$
performs
$$
\Omega\left( \frac{(1-\alpha-\delta) \cdot n}{e^\epsilon} \right)
$$
operations in expectation.
\end{theorem}
\begin{proof}[Proof of Theorem~\ref{thm:low_ir}]
Since IR has error probability at most $\alpha$,
$\Pr[B_j\in\ir(j)]\ge 1-\alpha$, for all $j\in [n]$.
By Lemma~\ref{lemma:pir1}, for all $j \ne i$,
$$
\Pr[B_j\in\ir(j)] \le e^\epsilon \Pr[B_j \in \ir(i)]+\delta.
$$
Equivalently, this means that $\Pr[B_j \in \ir(i)] \ge \frac{1-\alpha-\delta}{e^\epsilon}$. So,
$$
\E[|\ir(i)|] \ge \sum_{j\ne i}\Pr[B_j \in \ir(i)]\geq(n-1)\frac{1-\alpha-\delta}{e^\epsilon}
$$
that yields the theorem.
\end{proof}

For any constant error $\alpha > 0$ and typical privacy budgets of
$\epsilon = \Theta(1)$, errorless $\dpir$ schemes require
$\Omega(n)$ server operations.
However, it seems possible to bypass the first negative result with very
small error.
The above lower bound does not preclude the existence
of a $\edpir$ with $\epsilon = \Omega(\log n)$ and small
error probability $\alpha > 0$
that only requires a constant number of server operations.
In Section~\ref{sec:ir}, we present an
$\edpir$ scheme that uses $O(1)$ communication when $\epsilon = O(\log n)$.
It turns out this $\edpir$ scheme asymptotically matches
the lower bound for all values
of $\epsilon \ge 0$.

We also extend our negative results for
$\dpir$ in Section~\ref{sec:low_ir_mult} where
we present lower bounds when outsourcing storage
to multiple, non-colluding servers.

%% file: lower_ram.tex
\subsection{DP-RAM}
We now move to our negative results for $\dpram$.
Proving lower bounds for $\dpram$ is more challenging than $\dpir$
due to the private
setup phase and the client's private memory.
We can no longer directly bound the probability that blocks need
to be retrieved and, instead, examine the transcripts seen by the adversary.
Let us fix any transcript $\calT$ that has non-zero probability of
being viewed by the adversary on any query sequence $Q$ of length $l$.
We show that every other query sequence of length $l$ must also
induce $\calT$ with non-zero probability.
In fact, the probabilities
that any two fixed query sequences induce $\calT$ as the view of the
adversary are strong related by their Hamming distance and
the privacy budget, $\epsilon$,
as described in the following:
\begin{lemma}
\label{lem:aposteriori}
Let $\ram$ be a $\edpram$ scheme for any $\epsilon \ge 0$.
For every distribution $\calD$
on query sequences of length $l$, random variable $\dist \sim \calD$
and for any two fixed
sequences $Q_1, Q_2$ such that $\Pr[D = Q_2] > 0$, then
\begin{align*}
        \frac{\Pr[D = Q_1 \mid \ram(D) = \calT]}{\Pr[D = Q_2 \mid \ram(D) = \calT]} \ge e^{-\epsilon \cdot d(Q_1,Q_2)} \frac{\Pr[D = Q_1]}{\Pr[D = Q_2]}.
\end{align*}
\end{lemma}
\begin{proof}
By Bayes' law, we have
\begin{align*}
\Pr[&\dist=Q \mid \ram(\dist)=\trans]\\
&= \frac{\Pr[\ram(\dist)=\trans\mid \dist=Q]\Pr[\dist=Q]}{\Pr[\ram(\dist)=\trans]}\\
&\ge e^{-\epsilon\cdot d_H(Q, Q')}
   \frac{\Pr[\ram(\dist)=\trans\mid \dist=Q']\Pr[\dist=Q]}{\Pr[\ram(\dist)=\trans]}\\
&= e^{-\epsilon\cdot d_H(Q, Q')}\frac{\Pr[\dist=Q'\mid\ram(\dist)=\trans]\Pr[\dist=Q]}{\Pr[\dist=Q']}
\end{align*}
giving us our result.
\end{proof}

Next, we prove that if a query sequence has positive probability,
then it stays so even after seeing a transcript.
This holds for all {\em a priori} query sequences distributions $\calD$.

\begin{lemma}
\label{lem:deltaZero}
Let $\ram$ be an $\edpram$ and let $\trans$ be a transcript for
which there exists at least one query sequence $Q$ such that
$$\Pr[\ram(Q) = \trans] > 0.$$
Then for all distributions $\dist$ on the set of query sequences and for all
$Q'$ such that $\Pr[\dist=Q']>0$ it holds that
$$\Pr[\dist=Q'|\ram(\dist)=\trans]>0.$$
\end{lemma}
\begin{proof}
Let $Q$ be any query sequence
such that $\Pr[\dist=Q] >  0$ and $\Pr[\ram(\dist)=\trans\mid \dist=Q]>0$.
Assume by contradiction that $\Pr[\dist=Q']>0$ but
$\Pr[\dist=Q'\mid\ram(\dist)=\trans]=0$ for some query sequence $Q'$.
By Bayes' Law,
\begin{align*}
        0&=\Pr[\dist=Q'\mid\ram(\dist)=\trans]\\
	&=\frac{\Pr[\ram(\dist)=\trans\mid \dist=Q']\cdot\Pr[\dist=Q']}{\Pr[\ram(\dist)=\trans]}
\end{align*}
which implies that $\Pr[\ram(\dist)=\trans\mid \dist=Q']=0$.
On the other hand,
by $\epsilon$-differential privacy, we have that for all $q$ such that $\Pr[\dist=Q]>0$,
\begin{align*}
\Pr[&\ram(\dist)=\trans\mid \dist=Q]\\
&\leq e^{d_H(Q, Q')\epsilon}\cdot \Pr[\ram(\dist)=\trans\mid \dist=Q'] = 0
\end{align*}
providing our contradiction.
\end{proof}

We are now ready to prove our $\dpram$ lower bound.
For the sake of the completeness, we suppose
that the $\dpram$ scheme only retrieves the desired block with
probability at least $1 - \alpha$ based only on internal randomness.
We show that:

\begin{theorem}
\label{thm:low_ram}
If $\ram$ is an $\edpram$ scheme in the balls and bins model
with error probability $\alpha \ge 0$ then $\ram$ performs
$$
\Omega\left( \log_c \left(\frac{(1-\alpha) \cdot n}{e^\epsilon}\right) \right)
$$
expected amortized operations per query
when the client has storage for $c$ balls.
\end{theorem}
\begin{proof}
Fix $Q = (q_1,\ldots,q_l)$ to be any sequence of queries of length $l>c$.
Then, $\client_c$ can only perform upload/\allowbreak download operations
with $\server_M$.
Denote the expected amortized bandwidth by $k$.
The server locations of all operations are given by the transcript, $\ram(q)$.
However, which of the $c$ client locations used for each operation
remains hidden.
For each download and upload, there are $c$ possible execution paths, one
path for each of the $c$ client memory locations.
Also, after each operation, $\client_c$ may use the data stored in any
of the $c$ client memory locations to answer a query.
Altogether, at most $kl$ blocks of bandwidth are used for all $l$ queries.
Therefore, there are at most $c^{2kl}$ different sequences of blocks
that may be returned.
Since $\ram$ can only fail with $\alpha$ probability for each query,
it must satisfy at least $(1-\alpha)^l n^l$ different access patterns.
So, $c^{2kl} \ge (1-\alpha)^l n^l$.

Let $\trans$ be any transcript with positive probability and consider
the uniform distribution $U$ over the set $[n]^l$ of query sequences of length $l$.
By Lemma~\ref{lem:deltaZero},
there exists $Q\in [n]^l$ such that
$$\Pr[U=Q\mid\ram(U)=\trans]\ge\frac{1}{n^l}\geq\frac{(1-\alpha)^l}{c^{2kl}}.$$
For every $Q'\ne Q$, by Lemma~\ref{lem:aposteriori}, we have
$$
\Pr[U=Q'\mid\ram(U)=\trans]
\ge e^{-l\epsilon}\frac{(1-\alpha)^l}{ c^{2kl} }.$$
Finally, we know that
\begin{align*}
1 &\geq \sum\limits_{Q'\ne Q} \Pr[U=Q'\mid \ram(U)=\trans]
\ge e^{-l\epsilon}(n^l - 1)\frac{(1-\alpha)^l}{ c^{2kl} }
\end{align*}
completing the proof.
\end{proof}

Unlike $\dpir$, we do not need a non-zero error probability to
get a constant overhead scheme.
From the above lower bound, it seems like the best one can achieve
is a perfectly correct $\edpram$ with $\epsilon = \Theta(\log n)$
that only requires $O(1)$ overhead with small client storage.
In Section~\ref{sec:ram}, we show that such a scheme does exist.

%% file: lower_disc.tex
\paragraph{Discussion about lower bounds.}
It has been shown that previous ORAM lower bounds come with many cavaets.
The first lower bound~\cite{GO96} in the balls and bins model
was for statistical security as pointed out in~\cite{BN16}.
The assumption of statistical security is troublesome due to the fact that
all ORAM schemes require the use of encryption which only provides computational
security. To abstract away this issue, it is assumed that records
were opaque balls and were hidden even against computationally unbounded
adversaries. However, this abstraction was still not sufficient since the
majority of ORAM schemes still required the use of pseudorandom functions.
Furthermore, the $\dpram$ constructions of~\cite{WCM16} and~\cite{CCM18}
require encryption and are burdened by the same cavaets.
We explain why these issues do not apply to our results.
In our work, we consider typical differential privacy notions (not the computational variants described in~\cite{MPR09}) where adversaries are computationally
unbounded. Our $\dpram$ scheme in Section~\ref{sec:ram} when only allowing
retrievals
does not require encryption and provides differentially private access
to public data against computationally unbounded adversaries.
All ORAM schemes still require encryption even when only retrievals are permitted.
In the case that we wish to store encrypted data or protect overwrites
in our $\dpram$ scheme,
we must apply the abstraction of opaque balls to go around encryption.
With this abstraction, our $\dpram$ scheme is differentially private.
Recent results in~\cite{LN18} present $\Omega(\log (n/c))$ lower bounds
for ORAM against computational adversaries with
passive servers and general storage schemes.
Furthermore, work in~\cite{PY18} extends the lower bound for
$\eddpram$ where $\epsilon = O(1)$ and $\delta \le 1/3$.
However, the lower bounds in~\cite{PY18} have an exponentially worse
dependency on $\epsilon$ compared to our results.

%% file: strawman.tex
\section{An Insecure Construction}
\label{sec:insecure}
Before presenting our constructions,
we consider a simple and tempting, but insecure, construction.
The lower bounds presented in Section~\ref{sec:low} show that
the best privacy any constant overhead storage primitive can achieve
will be $\epsilon = \Theta(\log n)$. With such weak privacy requirements,
it might seem very easy to construct these primitives at first.
We caution that schemes with these weak privacy requirements must be
constructed carefully as slight variants of our later schemes
could also end up being insecure. To our knowledge, our schemes
are the simplest constructions that achieve $\epsilon = \Theta(\log n)$
differential privacy and small overhead.

The main idea of the strawman solution derives from the fact that
$\epsilon = \Theta(\log n)$. As a result, the desired block
should be queried with probability a multiplicative factor of $\poly(n)$
larger to compared to any other block.
To achieve this, one could query the desired block with probability $1$
and all other blocks with probability $1/n$. This scheme would have
$O(1)$ bandwidth in expectation, perfect correctness and no
client storage requirements.
However, we show that this scheme
is really an $\eddpir$ with $\epsilon = \Theta(\log n)$ and $\delta = (n-1)/n$.
Denote the above scheme by $\ir$ and $\ir(i)$ the set of blocks
returned when querying for $i$.
Pick any two queries $i \ne j$.
Note, $\Pr[B_i \notin \ir(i)] = 0$ and $\Pr[B_i \notin \ir(j)] = (n-1)/n$.
Then, $(n-1)/n = \Pr[B_i \notin \ir(j)] \le e^\epsilon \Pr[B_i \notin \ir(i)] + \delta$ which means that $\delta \ge (n-1)/n$.
Therefore, the above scheme would not provide as much privacy as possible
as $\delta$ approaches $1$ as $n$ increases.

We use the above strawman to show that attentiveness and rigor are required when
constructing differentially private storage schemes even with such
weak privacy guarantees.
The schemes that we will present for $\dpir$, $\dpram$ and $\dpkvs$
are the simplest algorithms that, to our knowledge,
achieve our desired privacy of $\epsilon = \Theta(\log n)$
and small overhead in efficiency.

%% file: ir.tex
\section{DP-IR Construction}
\label{sec:ir}
For errorless $\dpir$, the asymptotically optimal balls and bins algorithms
is required to download the entire database regardless of the privacy budget.
With active servers that can perform computation, servers must
perform $n$ operations which is identical to PIR.
Therefore, there is no reason to use differentially private access
when oblivious access has the same efficiency.

We now move to the more interesting case of $\dpir$ with errors.
By introducing a small amount of error $\alpha > 0$, we hope to find
algorithms more efficient than PIR for larger privacy budgets where
our lower bounds from Section~\ref{sec:low} no longer hold.
We show that the simplest algorithm ends up being optimal.
In particular, the client will download the desired block as well as several
other blocks simultaneously. The hope is that the adversary cannot determine
the real retrieval from all the fake retrievals.
In addition, with probability $\alpha$, we only perform fake retrievals
and error. The pseudocode of this scheme can be found in
Appendix~\ref{sec:dpir_code}.
The proof of the following theorem is found in Appendix~\ref{dpir-proof}

\begin{theorem}
\label{thm:dpir}
For any $\epsilon \ge 0$, there exists a $\edpir$
that returns $O(n/e^\epsilon)$
blocks for any constant error probability $\alpha > 0$.
\end{theorem}

The above upper bound asymptotically matches the
lower bound of Theorem~\ref{thm:low_ir} for all values of $\epsilon \ge 0$.
Furthemore, by fixing the privacy budget to be
$\epsilon = \Theta(\log n)$, we can achieve
a constant overhead $\dpir$ scheme with the best privacy according
to Theorem~\ref{thm:low_ir}.

%% file: ram.tex
\section{DP-RAM Construction}
\label{sec:ram}

In this section, we give an errorless construction $\oramC$
supporting both retrieval and overwrite operations.
Before describing our $\dpram$ scheme, we note that
one could also use the $\dpir$ scheme from Section~\ref{sec:ir}
as a $\dpram$ scheme without any client storage requirements.
However, this $\dpir$ scheme has non-zero error probabilities
which is inherently unavoidable for $\dpir$
due to our lower bounds in Section~\ref{sec:low}.
On the other hand, our $\dpram$ lower bounds do not preclude the existence
of a perfectly correct $\dpram$ scheme with $\epsilon = \Theta(\log n)$.
In this section, we work towards constructing such a $\dpram$ scheme.

Our scheme will require the client to store some records
in a local {\em stash}.
Our $\oramC$ scheme is parameterized by a probability $p$
describing the independent probability that
each record is stored in the stash.
We assume that
$(\Enc,\Dec)$ is an IND-CPA symmetric-key encryption scheme.
The server's storage will consist of an array, $A$, of $n$ records.

The setup phase of $\oramC$ will consist of populating the server-stored array $A$.
For security parameter $\lambda$,
key $K\rightarrow\{0,1\}^\lambda$ is randomly selected.
$A$ is initialized
by setting $A[i]=\Enc(K, B_i)$ for $i \in [n]$.
The stash
is initialized by independently selecting each
record to be in the stash with probability $p$.
In addition, the client keeps the key $K$ in local storage.

The querying (either retrieval or overwrite) for a record $B_i$ consists of
two phases:
the {\em download} phase followed by the {\em overwrite} phase.
In the download phase,
the client looks for $B_i$ in the stash.
If $B_i$ is found, then $B_i$ is removed and returned.
The client asks the server for $A[j]$,
with $j$ chosen uniformly at random from $[n]$.
If, instead, $B_i$ is not in the stash,
the client asks the server for $A[i]$ that contains an encryption of $B_i$.
If the client is performing a write operation,
then $B_i$ is updated with the new version.
At this point, the client holds the current version of $B_i$.

In the overwrite phase, the current version of $B_i$ is added to the
stash with probability $p$.
If $B_i$ is stored in the stash, then another record is randomly selected,
downloaded from the server,
decrypted and then re-encrypted with fresh randomness
and uploaded to the server.
If $B_i$ is not stored in the stash, then the client
asks the server for $A[i]$, discards the record received
and then uploads to $A[i]$ a freshly computed ciphertext carrying the
current version of $B_i$.

The pseudocode of the algorithms are presented in
Appendix~\ref{sec:ram_code}.
We note that while the above algorithm is simple, the analysis of privacy
is quite complicated.
We show the following about this scheme when $p \le \Phi(n)/n$ for
any $\Phi(n) = \omega(\log n)$:
\begin{theorem}
\label{thm:up_ram}
	There exists an $O(\log n)\text{-}\mathsf{DP}\text{-}\ram$ that
	returns $O(1)$ blocks.
	For any function $\Phi(n) = \omega(\log n)$, the client stores
	$\Phi(n)$ blocks
	of client storage except with probability $\negl(n)$.
\end{theorem}

This scheme is, essentially, the best privacy that can be achieved
by an errorless $\dpram$ scheme with constant overhead
according to Theorem~\ref{thm:low_ram}.
Since we use encryption, our $\dpram$ satisfies computational
differentially privacy using simulators (see SIM-CDP in~\cite{MPR09}).
Our proof uses a simulator that replaces all encryptions of records
with randomly generated contents.

\noindent{\em Discussion about encryption.} The above $\dpram$ scheme assumes that both
record retrievals and overwrites are permitted. To hide whether queries
are overwrites or retrievals, any $\dpram$ scheme must use encryption.
In the case that we wish to only permit retrievals, we note that
the above $\dpram$ scheme no longer requires encryption and can provide
differentially private access to public data without computational assumptions.
In particular, the entire overwrite phase may be skipped.
Retrieval-only $\dpram$ and $\dpir$ only differ by
their requirements on the client's state.
Therefore, using client state is another way to bypass the strong
lower bound shown in Theorem~\ref{thm:low_ir_no_error} for errorless $\dpir$ schemes.

%% file: ram_proof.tex
\subsection{Roadmap of the Proof}
To start, we show that if we choose that $p \le c/n$ where
$c = \omega(\log n)$, then the client will store at most $O(c)$
blocks. The proof is an application of Chernoff Bounds and postponed
to Appendix~\ref{sec:ram_proofs_app}.

The technical crux of the privacy analysis lies in bounding the following
ratio,
for every transcript $\calT$ seen by the adversary,
and for every two neighboring query sequences $Q$ and $Q'$,
\begin{align}\label{eqn:ratio}
\frac{\Pr[\ram(Q')=\calT]}{\Pr[\ram(Q)=\calT]}
\end{align}
where $\Pr[\ram(Q)=\calT]$ denotes the probability that $\ram$ executing
on $Q$ produces the transcript $\calT$.
In the first step,
we show that the transcript of the adversary at a single query on $B_q$
is only dependent on the most recent query that also queries
for $B_q$. In other words, the ratio is upper bounded by
$$
\prod_{i=1}^l \frac{\Pr[\ram_i(Q')=\calT_i\mid\ram_{\prev(Q', i)}(Q')=\calT_{\prev(Q',i)}]\ }{\Pr[\ram_i(Q)=\calT_i\mid\ram_{\prev(Q, i)}(Q)=\calT_{\prev(Q, i)}]}
$$
where $\prev(Q, i)$ is the most recent query for block $B_{q_i}$
before the $i$-th query.
The second step consists of
giving an upper bound that holds for each factor of the product in
the above equation.
The third step shows that the upper bound computed in the second step is too pessimistic and that,
all factors, except for $3$, in the product in the right hand side of the
above equation are $1$
when $Q$ and $Q'$ differ in exactly one position.

We now proceed to define notation and terminology used throughout
our proof.
First, we observe that the transcript of an execution only includes the
ciphertexts of the blocks and not the actual content of the blocks transferred.
Assuming IND-CPA of the underlying encryption scheme $(\Enc,\Dec)$,
it is straightforward to prove that,
for a sequence $Q=(q_1,\ldots,q_l)$ of $l$ queries,
the transcript generated by $Q$ for blocks $B_1,\ldots,B_n$
is indistinguishable from the transcript generated by the same sequence $Q$
for $n$ blocks that are $\bzero$.
For this reason, we shall not consider the ciphertexts of the blocks as part of the
transcript and consider a transcript $\calT$ for query sequence $Q$ as a sequence
$\calT=((d_1,o_1),\ldots,(d_l,o_l))$
of $l$ pairs $\calT_j=(o_j,d_j)$ of indices of the blocks that are accessed during the download phase and the overwrite phase.
We set $\calT_{[j]}=((d_1,o_1),\ldots,(d_j,o_j)).$
We define
$\ram^D_j(Q)$ and $\ram^O_j(Q)$ to be the random variables of the indices of the download and of the overwrite block
of the $j$-th query for all $j \in [l]$.
Also, for $S\subseteq [l]$ we define $\ram_S(Q)$
to be the set of random variables
$\{\ram^D_j(Q), \ram^O_j(Q)\}_{j \in S}$. Also, we set $\ram(Q):=\ram_{[n]}(Q)$.
It turns out to be convenient to extend the random variable
$\ram$ for indices $\pm\infty$ by setting
$\ram_{\pm\infty}^D(Q)$ and $\ram^O_{\pm\infty}(Q)$ to be random variables
that give probability $1$ to $\perp$.

We also let $\prev(Q,j)$ be the index of the most recent {\em previous} query
to block $B_{q_j}$ that happened before the $j$-th query of $Q$;
that is, $\prev(Q,j)=\max\{i<j:q_i=q_j\}$.
If query $q_j$ is the first query of
sequence $Q$ that asks for block $B_{q_j}$, then $\prev(Q,j)=+\infty$.
Similarly, we define $\nxt(Q,j)$ to be the index of the nearest
{\em next} query for block $B_{q_j}$; that is, 
$\nxt(Q,j) = \min\{i>j:q_i=q_j\}$.
If the $j$-th query of $Q$ is the last query of $Q$ to ask for block
$B_{q_j}$, then $\nxt(Q,j)=-\infty$.

\subsection{Step I: Reducing dependencies}
\label{sec:stepone}
The next lemmas outline the dependencies of the random variables of the
download blocks and of the overwrite blocks.
We start by showing that the overwrite block of each query is independent of
all previous history
and only depends on the current query.

\begin{lemma}
\label{lemma:overwrite}
For every query sequence $Q$ of length $l$,
for every transcript $\calT=((d_1,o_1),\ldots,(d_l,o_l))$ and
for every $j\leq l$
\begin{align*}
        \Pr[&\ram^O_j(Q)=o_j\mid\ram_{[j-1]}(Q)=\calT_{[j-1]}\wedge \ram^D_j(Q)= d_j]\\
        &= \Pr[\ram^O_j(Q)=o_j].
\end{align*}
Moreover, if $Q$ and $Q'$ are two sequences with $q_j=q'_j$, then the distributions
$\ram^O_j(Q)$ and $\ram^O_j(Q')$ coincide.
\end{lemma}
\begin{proof}
The lemma follows by observing that the distribution of the $j$-th overwrite block
depends only on whether $B_{q_j}$ is added to the block stash during the
$j$-th overwrite phase which, in turn, depends only on $q_j$ and the random value of $r$
in the overwrite phase of the $j$-th query.
\end{proof}

\begin{lemma}
\label{lemma:download}
For every query sequence $Q$ of length $l$,
for every transcript $\calT=((d_1,o_1),\ldots,(d_l,o_l))$ and
for every $j\leq l$,
\begin{align*}
  \Pr&[\ram^D_j(Q)=d_j\mid\ram_{[j-1]}(Q)= \calT_{[j-1]}]=\\
&\Pr[\ram^D_j(Q)=d_j\mid \ram^O_{\prev(Q,j)}=o_{\prev(Q,j)}].
\end{align*}
Moreover, if $Q$ and $Q'$ are two sequences with $q_j=q'_j$ and $\prev(Q,j)=\prev(Q',j)$ then for all $d\in[n]$ and
for all $o\in[n]\cup\{\perp\}$ such that $\Pr[\ram^O_{\prev(Q,j)}=o]>0$,
\begin{align*}
	\Pr[&\ram^D_j(Q)=d\mid\ram^O_{\prev(Q,j)}=o]\\
	&=\Pr[\ram^D_j(Q')=d\mid\ram^O_{\prev(Q',j)}=o].
\end{align*}
\end{lemma}
\begin{proof}
The $j$-th download block of sequence $Q$
depends on whether $B_{q_j}$ is in the block stash at the start of the
$j$-th download phase.
We distinguish two cases.

If $\prev(Q,j)=+\infty$, then block $B_{q_j}$ has not been queried in the first
$j-1$ queries and its probability of being found in the stash at the start of
$j$-th download phase is equal
to the probability of being placed in the stash by $\ram.\!\Init$ which,
obviously, is independent from the previous history.

Suppose instead that $\prev(Q,j)\in[j-1]$. Then the probability that $B_{q_j}$
is in the stash at the beginning of the $j$-th query depends only on
the random variable, $\ram^O_{\prev(Q,j)}(Q)$, of the overwrite block of query
$\prev(Q,j)$.
\end{proof}

The first part of Lemma~\ref{lemma:overwrite} and the first part of Lemma~\ref{lemma:download} imply
{\small
\begin{align*}
&\Pr[\ram_j(Q)=(d_j,o_j)\mid\ram_{[j-1]}(Q)=\calT_{[j-1]}]=\\
     &\Pr[\ram_j^D(Q)=d_j     \mid\ram_{[j-1]}(Q)=\calT_{[j-1]}] \cdot\\
     &\quad\Pr[\ram_j^O(Q)=o_j
  \mid\ram_{[j-1]}(Q)=\calT_{[j-1]}\wedge\ram_j^D(Q)=d_j]=\\
  &\Pr[\ram_j^D(Q)=d_j     \mid\ram_{\prev(Q,j)}(Q)=o_{\prev(Q,j)}]\cdot
  \Pr[\ram_j^O(Q)=o_j].\\
\end{align*}
}
Thus, we can write
{\small
\begin{align*}
&\frac{\Pr[\ram(Q')=\calT]}{\Pr[\ram(Q)=\calT]}=
        \prod_{j=1}^l
            \frac{\Pr[\ram_j^O(Q )=o_j]}%
              {\Pr[\ram_j^O(Q')=o_j]}\\
	&
	\qquad\times\prod_{j=1}^l
  \frac{\Pr[\ram_j^D(Q')=d_j\!\mid\!\ram_{\prev(Q',j)}(Q )=o_{\prev(Q',j)}]}%
       {\Pr[\ram_j^D(Q )=d_j\!\mid\!\ram_{\prev(Q ,j)}(Q')=o_{\prev(Q ,j)}]}
\end{align*}}

\subsection{Step II: Upper bounding factors}
\label{sec:steptwo}
In the next two lemmata,
we give an upper bound on the contribution of each $j\in[l]$ to the product in
the equation above.

\begin{lemma}
\label{lemma:download_prob}
Let $Q$ and $Q'$ be two query sequences of length $l$.
For every transcript $\calT=((d_1,o_1),\ldots,(d_l,o_l))$ and every $j\in[l]$
$$
\frac{\Pr[\ram^D_j(Q')=d_j\mid\ram^O_{\prev(Q',j)}(Q')=o_{\prev(Q',j)}]}%
	{\Pr[\ram^D_j(Q) =d_j\mid\ram^O_{\prev(Q ,j)}(Q )=o_{\prev(Q ,j)}]} \le \frac{n^2}{p}.
$$
\end{lemma}

\begin{lemma}
\label{lemma:overwrite_prob}
Let $Q$ and $Q'$ be two query sequences of length $l$.
For every transcript $\calT=((d_1,o_1),\ldots,(d_l,o_l))$ and every $j\in[l]$
	$$\frac{\Pr[\ram^O_j(Q')=o_j]}{\Pr[\ram^O_j(Q)=o_j]} \le \frac{n}{p}.$$
\end{lemma}

Both of these lemmata consider the various cases that can occur.
As the proofs are case analysis that do not provide better intuition to
the problem,
we postpone them to Section~\ref{sec:ram_proofs_app}.

\subsection{Step III: Identifying the many good cases}
The bounds given by Lemma~\ref{lemma:download_prob}
and Lemma~\ref{lemma:overwrite_prob}
would give an $n^{O(l)}$ upper bound on the ratio in
Equation~\ref{eqn:ratio}. This is a very weak bound as it depends on the length $l$ of the sequences.
In this section, we tighten the upper bound to $n^{O(1)}$ which is instrumental to prove that $\ram$ is private with $\epsilon = O(\log n)$.
Specifically, the
next lemma gives sufficient conditions under which the ratio is actually $1$
and then we show that,
if the two sequences only differ in one position, then there are only three values of $j$ for which
the conditions are not satisfied and for those position we use the upper bound of the previous section.

The following lemma follows directly from the second parts of Lemma~\ref{lemma:overwrite} and~\ref{lemma:download}.
\begin{lemma}
\label{lemma:same}
For any two sequences $Q$ and $Q'$ of the same length,
every transcript $\calT=((d_1,o_1),\ldots,(d_l,o_l))$
and every $j \in [l]$ with $\prev(Q,j)=\prev(Q',j)$ and $q_j=q'_j$,
	{\small
\begin{align*}
&\Pr[\ram^O_j(Q)=o_j]\\
&\qquad\cdot\Pr[\ram^D_j(Q)=d_j\mid\ram^O_{\prev(Q,j)}(Q)=o_{\prev(Q,j)}]=\\
&\Pr[\ram^O_j(Q')=o_j]\\
&\qquad\Pr[\ram^D_j(Q')=d_j\mid\ram^O_{\prev(Q',j)}(Q')=o_{\prev(Q',j)}].
\end{align*}}
\end{lemma}

The lemma above says that the distributions of the transcripts associated with
two query sequences $Q$ and $Q'$ may differ only at indices $j$ for which $\prev(Q,j)\ne \prev(Q',j)$ or
$q_j\ne q'_j$.  The next lemma identifies the indices $j$ for which this happens when $Q$ and $Q'$ differ
in exactly one position.

\begin{lemma}
\label{lemma:different}
Let $Q$ and $Q'$ be two query sequences of length $l$ differing only at position $k \in [l]$.
If $j \notin \{k,\nxt(Q,k),\allowbreak\nxt(Q',k)\}$, then $\prev(Q,j)=\prev(Q,j')$.
\end{lemma}
\begin{proof}
For $j<k$, we have $(q_1,\ldots,q_j)=(q'_1,\ldots,q'_j)$ and thus $\prev(Q,j)=\prev(Q',j)$.
For $j>k$ such that $q_j\ne q_k,q'_k$ we have $q_j=q'_j$ and $\prev(Q,j)=\prev(Q',j)$.
Next consider the indices $j_1<\ldots<j_l$ such that $q_k=q_{j_1}=\ldots=q_{j_l}$. Clearly,
if $l>0$ then $j_1=\nxt(Q,k)$ and, for $i=2,\ldots,l$, we have $q_{j_i}=q'_{j_i}$ and $\prev(Q,j_i)=\prev(Q',j_i)$.
A similar argument applies for the indices $j>k$ such $q'_j=q'_k$ thus completing the proof of the theorem.
\end{proof}

\subsection{Wrapping up the proof}
\begin{proof}[Proof of Theorem~\ref{thm:up_ram}]
Bounds on bandwidth and server storage are obvious and the one on
client storage follows from Lemma~\ref{lemma:oram_store}.

Consider sequences $Q$ and $Q'$ of length $l$
that differ only in position $k \in [l]$.
Let $\calT=((d_1,o_1), \ldots, (d_l,o_l))$ be transcript.
By Lemma~\ref{lemma:same} and~\ref{lemma:different},
we obtain that the ratio
	$\frac{\Pr[\ram(Q')=\calT]}{\Pr[\ram(Q) =\calT]}$
is
{
\small
$$
\prod_{j\in S}\frac
     {\Pr[\ram_j(Q')=\calT_j\!\mid\!\ram_{_{[j-1]}}(Q')=\calT_{_{[j-1]}}]}
     {\Pr[\ram_j(Q )=\calT_j\!\mid\!\ram_{_{[j-1]}}(Q )=\calT_{_{[j-1]}}]}
$$
}
where $S = \{k, \nxt(Q, k), \nxt(Q', k)\}$.
By Lemma~\ref{lemma:download_prob} and~\ref{lemma:overwrite_prob},
$$
	\frac{\Pr[\ram(Q')=\calT]}{\Pr[\ram(Q)=\calT]} = \left(\frac{n}{p}\right)^{O(1)}.
$$
Since, the above holds for any single transcript, we get that the ratio holds
for any set of transcripts. This implies that $\epsilon = O(\log n)$.
\end{proof}

%% file: kvs.tex
\section{DP-KVS Construction}
\label{sec:kvs}

In this section, we present $\dpkvs$, our $\kvs$ construction
with privacy budget $\epsilon = O(\log n)$ and small overhead.
Recall that a $\kvs$ is an extension of $\ram$
where each of the $n$ blocks is identified by a unique key
taken from a possibly large universe of keys $U$. In contrast,
a block in $\ram$ is uniquely identified by an integer in $[n]$.
One could use the $\dpram$ construction from Section~\ref{sec:ram},
which results in server storage on the order of $|U| \gg n$.
For efficiency, we would like a $\dpkvs$ scheme that stores $O(n)$ blocks
on the server.

Our approach to constructing a $\dpkvs$ scheme consists of two steps.
First, we show that we can construct a $\dpkvs$ scheme using
a $\dpram$ scheme and a {\em mapping scheme} which associates keys in
the universe $U$ to subsets of server storage. Next,
we present an efficient mapping scheme by constructing a non-trivial
{\em oblivious} variant of two-choice hashing~\cite{Mitzenmacher01}
that uses $O(n)$ server storage and which may be of independent interest.
Our oblivious two-choice hashing variant hides the
number of real items that are stored in each bin at any point in time.
Finally, combined with our $\dpram$ scheme of Section~\ref{sec:ram},
we present a $\eddpkvs$ with $\epsilon = O(\log n)$ and $\delta = \negl(n)$
using only $O(\log\log n)$ overhead. While non-constant, this is
exponentially better than any previous oblivious $\kvs$ scheme built from
ORAMs. Furthermore, for all practical sizes of $n$, $O(\log\log n)$ is
very small.

%% file: kvs_map_ram.tex
\subsection{Composing Mapping Schemes and DP-RAM}
\label{sec:kvs_map_ram}

In this section, we present a generic reduction of $\dpkvs$ to
a mapping scheme and our $\dpram$ from Section~\ref{sec:ram}.
First, we define mapping schemes. Afterwards, we present the reduction.

A mapping scheme is defined as the tuple $(\Pi, \calS)$ where $\Pi$
is the {\em mapping function} and $\calS$ is the {\em storing algorithm}.
To store $n$ items each uniquely identified by an key from the
universe $U$, a mapping scheme arranges the server storage
into $b(n)$ buckets each consisting of at most $s(n)$ blocks.
Each bucket is uniquely identified by an index from $[b(n)]$.
We note that buckets are not necessarily disjoint and the total number
of blocks may be much smaller than $b(n) \cdot s(n)$.
In addition, the mapping scheme assumes that the client will hold
a {\em mapping stash} which will contain at most $c(n)$ blocks except
with probability $\negl(n)$.
For each item $u \in U$, the mapping function maps $u$ to a subset
of at most $s(n)$ buckets defined by $\Pi(u) \subseteq [b(n)]$.
For convenience, we denote $k(n) := \max_{u \in U}|\Pi(u)|$.
When inserting a new item with identifier $u$,
the storing algorithm $\calS$ determines whether $u$ is placed
into a bucket of $\Pi(u)$ or the mapping stash according to the
sizes of buckets in $\Pi(u)$.
We now show how to use mapping schemes to construct a $\dpkvs$.

Our $\dpkvs$ scheme works as follows. We construct the server
storage into $b(n)$ buckets as described by the mapping scheme.
We build a $\dpram$ to be able to query and update the $b(n)$ buckets.
In Appendix~\ref{sec:general}, we show that our $\dpram$ construction
from Section~\ref{sec:ram}
remains secure and efficient when querying possibly overlapping buckets
with minor modifications.
When querying our $\dpkvs$ for a key $u \in U$,
we perform $k(n)$
$\dpram$ queries to retrieve the $s(n)$ blocks from
each bucket in $\Pi(u)$.
If $|\Pi(u)|<k(n)$,
we pick random buckets to pad $\Pi(u)$ to size $k(n)$.
Now, we are guaranteed that if $u$ exists in $\dpkvs$, it appears in
a bucket of $\Pi(u)$ or in the mapping stash and can be thus returned.
To update an existing key, we can simply update either the bucket or the mapping
stash that contains the block associated to key $u$.
For insertion, we can execute the storing algorithm $\calS$ as all the contents of
buckets $\Pi(u)$ and the mapping stash are available to the client.
$\calS$ determines the insertion location for the block associated to key $u$.
Finally, we execute $k(n)$ $\dpram$ updates to all buckets in $\Pi(u)$.
Only the bucket containing the block associated with key $u$ is updated. The
other buckets will perform fake updates where the contents remain unchanged.
For read operations, none of the contents of buckets in $\Pi(u)$ will be changed.
We prove the following theorem about our $\dpkvs$ construction.

\begin{theorem}
\label{thm:kvs_map_ram}
For $n$ blocks,
the above $\kvs$ scheme is an $\edpkvs$ with $\epsilon=O(k(n) \cdot \log n)$
	that returns at most $O(k(n) \cdot s(n))$ blocks.
For any function $\Phi(n) = \omega(\log n)$,
	the client stores $O(s(n) \cdot \Phi(n)+c(n))$ blocks
of storage except with probability $\negl(n)$.
\end{theorem}
\begin{proof}
The bounds on bandwidth and server storage follow from
the mapping scheme properties. The client storage bound follows
from Lemma~\ref{lemma:oram_store} and the mapping scheme properties.
Each query results in at most $2\cdot k(n)$ queries over the repertoire $\Sigma$ of the $n$ buckets.
By the composition theorem and the discussion in
Appendix~\ref{sec:general},
we obtain that $\epsilon = O(k(n) \cdot \log n)$.
\end{proof}

%% file: kvs_map.tex
\subsection{Oblivious Two-Choice Hashing}

Before presenting our new mapping scheme, we revisit the
two-choice hashing scheme~\cite{Mitzenmacher01}
(see Section~\ref{sec:two_choices} for more details).
This scheme considers $n$ buckets to store up to $n$ keys.
The mapping function $\Pi: U \rightarrow [n]$ sets $\Pi(u)$
as $k(n):=2$ independently and uniformly at random chosen buckets.
Typically, $\Pi$ is succintly represented using two keys,
$\key_1, \key_2$, of a pseudorandom function $F$ and by
$\Pi(u):=\{F(\key_1, u), F(\key_2, u)\}$.
The storing algorithm $\calS$ for $u \in U$, checks which of the
two buckets in $\Pi(u)$ is less loaded and places $u$ into the
less loaded bucket.
There are several different proofs that show that the largest bucket
will contain at most $s(n):=O(\log\log n)$ items except with probability $\negl(n)$.

Unfortunately, we are unable to use two-choice hashing directly
into our $\dpkvs$ scheme without incurring into a server storage blockup.
The $\dpkvs$ scheme from Section~\ref{sec:kvs_map_ram} requires
that all buckets are the same size for privacy.
The naive approach is to simply increase all buckets to the worst
case size which results in $O(n\log\log n)$ server storage.
Instead, we now present a variant of two-choice hashing
which will only use $O(n)$ server storage by arranging buckets
to share memory.

Our bucket arrangement is best described as $\Theta(n/\log n)$
identical binary trees, each with $\Theta(\log n)$ leaf nodes
and $\Theta(\log\log n)$ depth.
Leaf nodes are denoted as {\em height} 0 and the height increases
going towards the root.
This results in a total of $\Theta(n)$ nodes over all binary trees.
Each node in the tree will be able to store up to $t=\Theta(1)$ blocks.
Furthermore, we pick the binary trees such that there are exactly
$n$ leaf nodes overall.
Finally, there is a single root node that has $\Theta(n/\log n)$ children
corresponding to the roots of the $\Theta(n/\log n)$ binary tree roots.
We denote this node as the {\em super root}.
Unlike all other nodes, the super root is stored on the client.
We shall show that the probability that the super root holds more then $\Phi(n)$ blocks for
any $\Phi(n)=\omega(\log n)$ is negligible in $n$.

Each of the $n$ buckets is uniquely associated with a leaf node.
The memory locations of a bucket consist of all the blocks stored in the nodes on the unique path between
the leaf node and the super root. Therefore, a bucket consists of
$\Theta(\log\log n)$ server memory locations in addition to the memory locations in the super root.
We now describe our new storing algorithm, $\calS$, given this bucket arrangement.
When inserting $u$, $\calS$ places $u$ into the node
with minimal height (that is, closest to the leaf nodes)
in either of the buckets in $\Pi(u)$ with empty space. Note that $u$ might end up being stored in the super root.
If all the nodes of both buckets of $\Pi(u)$ are filled then the mapping scheme fails to store $u$.
We show that when inserting any set of at most $n$ keys,
if we limit the capacity of the super root to $\Phi(n)$, for some $\Phi(n)=\omega(\log n)$, the
above mapping scheme fails with probability $\negl(n)$.
The analysis adapts techniques from~\cite{ABK99}.

\begin{theorem}
\label{thm:kvs}
Let $\Phi(n)=\omega(\log n)$.
The probability that, when inserting $n$ keys, 
mapping scheme $\calS$ places more than $\Phi(n)$ keys into
the super root is $\negl(n)$.
\end{theorem}
\begin{proof}
Our key observation is that a block is stored in a level-($i$+1) slot
if and only
if the two selected buckets are filled up to level $i$.
If we denote by $H_i$ the number of filled node at level $i$, the
probability of selecting a bucket that is filled up to level $i$ is
$H_i\cdot 2^i/n$.
For a block to be allocated to a slot at level $i+1$, it must be the case that
both selected buckets are filled up to level $i$ which has probability $\left(\frac{H_i\cdot 2^i}{n}\right)^2$.
For convenience, we define the sequence $\beta_i$ by setting $\beta_0=\frac{n}{e\cdot 3^4}$ and
$\beta_{i+1}=\frac{e}{n}\cdot \beta_{i}^2\cdot 2^{2(i+1)}$, which we
will use later in the proof. First, we present a useful property about
the sequence and then we show that the probability that $H_i>\beta_i$ is $\negl(n)$.
The proof of Lemma~\ref{lemma:closed} can be found in Appendix~\ref{sec:kvs_proofs_app}.
\begin{lemma}
\label{lemma:closed}
For all $i \ge 0$,
        $$\beta_i=\frac{n}{e}\cdot\left(\frac{2}{3}\right)^{2^{i+2}}\left(\frac{1}{2}\right)^{2(i+2)}.$$
\end{lemma}

\begin{lemma}
\label{lemma:negligible}
If $\beta_{i}=\omega(\log n)$ then,
$\Pr[H_{i}>\beta_{i}] \le i/ n^{\omega(1)}$.
\end{lemma}
\begin{proof}
We proceed by induction on $i$.
The base case $i=0$ is established by setting $c>3^4\cdot e$.
We next upper bound the probability that after inserting $n$ blocks there are more than $\beta_{i+1}$ filled nodes at
level $i+1$, given that we start with $H_i$ nodes filled at level $i$.
To do so, we define $X_j$ to be that 0/1 random variable that is $1$ iff the $j$-th block ends up at level $i+1$.
Clearly, $\Pr[X_j=1]=\left(\frac{H_i\cdot 2^{i+1}}{n}\right)^2$ and, if $H_i\le \beta_i$,
$$\mu_{i+1}:=\E\left[\sum_j X_j\right]\leq \frac{\beta_i^2\cdot 2^{2(i+1)}}{n}=\frac{\beta_{i+1}}{e}.$$
We thus have
\begin{align*}
& \Pr[H_{i+1}>\beta_{i+1}\mid H_i\leq \beta_i]\leq \frac{\Pr[H_{i+1}>\beta_{i+1}]}{\Pr[H_i\leq \beta_i]}\\
=&\frac{\Pr[\sum_jX_j> e\cdot \mu_{i+1}]}{\Pr[H_i\leq \beta_i]}
\leq\frac{e^{-e\cdot\beta_{i+1}}}{\Pr[H_i\leq \beta_i]}
\end{align*}
where the last inequality follows by Theorem~\ref{thm:chernoff}.
By Lemma~\ref{lemma:closed}, the sequence $\beta_i$ decreases with $i$ and thus
$\beta_i\geq \beta_{i+1}=\omega(\log n)$.
The proof is then completed by observing that
\begin{align*}
        \Pr[&H_{i+1}>\beta_{i+1}]\\
&\le \Pr[H_{i+1}>\beta_{i+1}\mid H_i\le\beta_i]\cdot\Pr[H_i\le\beta_i]
+\Pr[H_i > \beta_i]\\
	&\le 1/ n^{\omega(1)} + i/n^{\omega(1)} \le (i+1)/n^{\omega(1)}
\end{align*}
using our hypothesis that $\Pr[H_i > \beta_i] = i/n^{\omega(1)}$.
\end{proof}

We now conclude the proof of Theorem~\ref{thm:kvs}.
Fix a function $\Phi(n)=\omega(\log n)$ and let $i^\star$
be the largest index such that $\beta_{i^\star}\geq\Phi(n)$.
By Lemma~\ref{lemma:closed}, we obtain that $i^\star=\Theta(\log\log n)$.
Following a reasoning similar to the one adopted in proof of Lemma~\ref{lemma:negligible}, we can prove
that the expected number of filled nodes at level $i^\star+1$ is at most $\beta_{i^\star+1}/e$, given
that no more than $\beta_{i^\star}$ nodes are filled at level $i^\star$.
By Lemma~\ref{lemma:negligible}, the condition
holds except with negligible probability.
From the definition of $\beta_{i^\star+1}$, we obtain that, for some constant $\alpha$ such that
$c\geq \alpha\cdot\beta_{i^\star+1}$,
the probability that more than $c$
nodes are filled at level $i^\star+1$ is inverse exponential in $c$ by Chernoff
Bounds, and thus $\negl(n)$.
\end{proof}

\subsection{Wrapping up the Proof}
We complete our $\dpkvs$ construction by observing that the mapping
scheme described above has $k(n)=2$, $s(n)=\Theta(\log\log n)$ and $c(n)=\Phi(n)$ for any $\Phi(n) = \omega(\log n)$.

\begin{theorem}
	The above $\kvs$ scheme is a $\edpkvs$ with $\epsilon = O(\log n)$
	that returns $O(\log\log n)$ blocks. The server uses $O(n)$ blocks of storage and,
	for any function $\Phi(n)=\omega(\log n)$, the probabilty that the client
	stores more than $O(\Phi(n) \cdot \log\log n)$ blocks is $\negl(n)$.
\end{theorem}

%% file: concl.tex
\section{Conclusions}
\label{sec:concl}
We consider privacy-preserving storage protocols with small overhead that could be implemented
with large-scale, frequently accessed storage infrastructures without
negatively impacting response times or resource costs. Our main question is
to find the best privacy that can be achieved by small overhead storage schemes.

We formulate our privacy notion using differentially private access,
which is a generalization
of the oblivious access provided by both ORAM and PIR.
We present strong evidence that the best privacy achievable by any constant overhead
differentially private storage schemes must have privacy budgets $\epsilon = \Omega(\log n)$.
For $\dpram$ and $\dpir$, we present constructions with asymptotically optimal
$\epsilon = \Theta(\log n)$ privacy budgets and $O(1)$ overhead.
For $\dpkvs$, we present a scheme with asymptotically optimal $\epsilon = \Theta(\log n)$
privacy budgets and only $O(\log\log n)$ overhead which is exponentially better than
previous constructions. Our $\dpkvs$ uses a novel, oblivious variant of two-choice hashing
that uses only $O(n)$ server storage that may be of independent interest.

Therefore, we answer that the best privacy achievable by privacy-preserving
storage systems with small overhead
is differentially private access with $\epsilon = \Theta(\log n)$.
On the other hand, any storage scheme
achieving stronger privacy most likely must incur non-trivial overhead
compared to plaintext access.

%% file: tools.tex
\section{Tools}
In this section we briefly review some tools that we use in the design and in the analysis of our constructions.


\subsection{Power of Two Choices}
\label{sec:two_choices}
The {\em power of two choices} concept was motivated from the classical
{\em balls and bins} problem. The balls and bins concepts considers
$n$ balls and $n$ bins. Each of the $n$ balls chooses a single bin independently
and uniformly at random. The {\em load} of a bin is the number of
balls that occupy the bin. It has been shown that with high probability,
the load of every bin does not exceed $O(\log n / \log\log n)$
~\cite{DR98}.
Consider the case where each of the $n$ balls now chooses two bins
independently and uniformly at random. The ball occupies
the least loaded of the two chosen bins. This slight alteration ensures
that with high probability, the load of each bin does not exceed
$O(\log\log n)$~\cite{Mitzenmacher01}.
This result demonstrates that having two choices significantly
improves bounds on the maximum load. Furthermore,
it turns out that increasing the
number of choices to $d \ge 3$ only improves the maximum load bounds
by a constant. It is important to note that the allocation of each
is chosen independently from the allocation of all other balls.

\begin{theorem}\label{thm:advDel}
At any time, the load of any bin produced by the power of two choices process exceeds
$O(\log\log n)$ with probability at most $1/n^{\Omega(\log\log n)}$.
\end{theorem}

\subsection{Chernoff Bound}
The next theorem gives a bound on the tails of a binomial distribution that we will use to analyze our constructions.
See~\cite{Mitzenmacher00} for a proof.
\begin{theorem}
\label{thm:chernoff}
Let $X_i$, for $i=1,\ldots,n$ be independent binary random variables with $\Pr[X_i=1]=p$ and let $\mu:=np$.
Then for every $t\geq \mu$, it holds that
$$\Pr\left[\sum_{i=1}^n X_i\geq t\right]\leq \frac{\mu^t}{t^t}\cdot e^{t-\mu}$$
and, in particular,
$$\Pr\left[\sum_{i=1}^n X_i\geq e\cdot\mu\right]\leq e^{-\mu}.$$
\end{theorem}

%% file: dpir-proof.tex
\section{Proof of Theorem~\ref{thm:dpir}}
\label{dpir-proof}
We remind the reader that, for $i\in[N]$, $\ir(i)$ is the random variable of the set
of blocks transferred by the server when the client wishes to access block $B_i$.
We next compute $\Pr[\ir(i)=\calT]$ for a subset $\calT$ of $K$ blocks.

\smallskip\noindent{\bf Case 1:} $B_i\in\calT$.
With probability $1-\alpha$, we know that
$B_i\in\calT$ and the remaining $K-1$ blocks are chosen uniformly at random.
On the other hand, with probability $\alpha$, for $B_i \in \calT$, we need
to choose $B_i$ as one of the $K$ blocks that are chosen randomly.
Therefore,
$$
\Pr[\ir(i) = \calT] = \frac{1-\alpha}{{N-1 \choose K-1}} + \frac{\alpha}{{N \choose K}}.
$$

\smallskip\noindent{\bf Case 2:} $B_i\notin\calT$.
In this case, we know with probability $(1-\alpha)$, $\calT$ is not possible.
With probability $\alpha$, we need to ensure that $B_i$ is not one of the
$K$ blocks chosen. So,
$$
\Pr[\ir(i) = \calT] = \frac{\alpha}{{N \choose K}}.
$$
Let $Q$ and $Q'$ be any two query sequences of length $L$ for which
$d_H(Q,Q')=1$. Since, the above algorithm is stateless, we know
that its behavior depends solely on the query index.
Therefore, it suffices to consider two different sequences queries of
length one,
$q \ne q' \in [N]$. If we choose a transcript $\calT$ such that
$B_{q}, B_{q'} \notin \calT$ or $B_{q}, B_{q'} \in \calT$, then
we see that $\Pr[\ir(q) = \calT] = \Pr[\ir(q') = \calT]$.
If, instead, $B_{q} \in \calT$ and $B_{q'} \notin \calT$ we have
\begin{align*}
\frac{\Pr[\ir(q) = \calT]}{\Pr[\ir(q') = \calT]}
        &\le \frac{(1-\alpha){N \choose K}}{\alpha{N - 1 \choose K - 1}} + 1
        \le \frac{(1-\alpha)N}{\alpha K} + 1
= e^\epsilon.
\end{align*}
Since this holds for any single transcript $\calT$, the same ratio
holds for any set of transcripts completing the proof.

%% file: ir_low_mult-app.tex
\section{Lower Bound for Multiple-Server DP-IR}
\label{sec:low_ir_mult}
We extend the single server lower bound to the multiple server IR model with
$D$ servers and $D_A$ adversarial servers.
If $D_A = D$, then the scenario collapses to the single adversarial server
case of Section~\ref{sec:low_ir}. In this section,
we assume that $D_A < D$ and there
is always at least one honest server.
For convenience, we define $t = \frac{D_A}{D}$ to be the
fraction of servers corrupted by $\adv$ where $0 < t < 1$.

The choice of adversarial servers is modeled as
a challenger $\ch$ with a $\ir$ protocol and an adversary $\adv$ with the
power to corrupt $t$ fraction of servers.
Using the knowledge of $\ir$ but not the internal randomness of $\ir$,
$\adv$ picks $t$ fraction of the $D$ servers to corrupt.
$\ch$ then runs $\ir$ and $\adv$ gets the transcript of downloads
sent by $\ir$ to the $t$ fraction of corrupted servers.
We denote the adversary's transcript on query sequence $Q$ as
$\ir^{\adv}(Q)$.
We now proceed to present our lower bound in the multi-server model.

\begin{theorem}
\label{thm:multiServerPIR}
If $\ir$ is a $(\epsilon,\delta)$-DP $D$-server IR in the balls and bins model
for any $\epsilon,\delta \ge 0$ and error probability
$\alpha < 1 - (\delta / t)$, then $\ir$ performs
$$
\Omega\left(\frac{((1-\alpha)t-\delta) \cdot n}{e^\epsilon}\right)
$$
expected operations.
\end{theorem}
\begin{proof}
We prove our lower bound for the adversary $\adv$
that corrupts $D_A\leq t D$ randomly and uniformly chosen servers.
We know
$$\Pr[B_i \in \ir(i)] \ge 1-\alpha.$$
If $B_i$ appears in $\ir^{\adv}(i)$,
$\Pr[B_i \in \ir^{\adv}(i)] \ge (1-\alpha)t$.
The rest of the proof follows identically to the proof of Theorem~\ref{thm:low_ir}.
\end{proof}

As a result, we show that the $\dpir$ scheme in~\cite{TDG16} is optimal
when $t$ is considered constant.

%% file: ram_leftover_proofs.tex
\section{Remaining DP-RAM Proofs}
\label{sec:ram_proofs_app}
In this section we give the proofs that were postponed
from Section~\ref{sec:ram}.
We start by showing that the client uses $O(c)$ memory except
with negligible probability for $c=\omega(\log n)$.

\begin{lemma}
\label{lemma:oram_store}
If $p \le \frac{c}{n}$ and $c=\omega(\log n)$, then $\ram$ stores
at most $O(c)$
blocks in the block stash on client storage at any point in time
except with negligible probability.
\end{lemma}
\begin{proof}
Pick any point in time. We note that each block $B_i$ has independent
probability $p$ of being stored in the client. We let $X_i = 1$
if and only if $B_i$ is stored on the client memory.
Let $X$ be the number of blocks stored in client memory, so
$X=X_1+\ldots+X_n$.
So,
we know that $\E[X]=\E[X_1 + \ldots + X_n]=pn \le c$.
By Chernoff Bounds, for any $\delta > 0$,
$$
\Pr[X > (1+\delta)c] \le \exp\left(\frac{-c\delta^2}{2+\delta}\right)
$$
which is $\negl(n)$ when $c = \omega(\log n)$.
The lemma follows, by a union bound over all points in time
(that is, the length of the queries which is polynomial in $n$).
\end{proof}

\begin{proof}[Proof of Lemma~\ref{lemma:download_prob}]
The numerator is trivially upper bounded by $1$ and the denominator is
lower bounded by the conjunctive probability.
We prove the lemma by showing that the latter is at least $p/n^2$.
To this aim, we consider
three cases depending on the values of $\prev(Q,j)$ and $o_{\prev(Q,j)}$.
Note that, if $\prev(Q,j)\ne +\infty$, $q_{\prev(Q,j)}=q_j$ by definition.

\smallskip\noindent{\bf Case 1:}
$\prev(Q,j)\ne+\infty$ and $o_{\prev(Q,j)}=q_j$.

Consider first the case $d_j=q_j$.
If during the overwrite phase of query $\prev(Q,j)$
block $B_{q_j}$ is stored in block stash (and this has probability $p$),
then $\ram^O_{\prev(Q,j)}(Q)=q_j$ with probability $1/n$. Then at the download phase of
query $j$, $B_{q_j}$ is found in block stash and therefore the probability that
$\ram^D_j(Q)=q_j$ is $1/n$.
On the other hand
if during the overwrite phase of query $\prev(Q,j)$
block $B_{q_j}$ is not stored in block stash (this has probability $1-p$),
then certainly $\ram^O_{\prev(Q,j)}(Q)=q_j$ and, since $B_{q_j}$ is not found in block stash
at the start of the download phase of query $j$,
$\ram^D_j(Q)=q_j$ with probability $1$. Altogether, we have
$\Pr[\ram^O_{\prev(Q,j)}(Q)=q_j\wedge \ram^D_j(Q)=q_j]
=1-p+1/n^3\geq 1/n^3.$
If, instead, $d_j\ne q_j$ then it means that $B_{q_j}$ must have been stored in the block stash
during the overwrite phase of query $\prev(Q,j)$ (this has probability $p$) and
then $\ram^O_{\prev(Q,j)}(Q)=q_j$ with probability $1/n$. Then, since $B_{q_j}$ is in stash at the
start of the download phase of query $j$
$\ram^D_j(Q)=d_j$ with probability $1/n$. Altogether, we have
$\Pr[\ram^O_{\prev(Q,j)}(Q)=q_j\wedge \ram^D_j(Q)=q_j]
=p/n^2.$

\smallskip\noindent{\bf Case 2.}
$\prev(Q,j)\ne+\infty$ and $o_{\prev(Q,j)}\ne q_j$.

If $o_{\prev(Q,j)}\ne q_j$ then it must be the case that, during the overwrite phase of query $\prev(Q,j)$,
block $B_{q_j}$ is stored in block stash (and this happens with probability $p$)
and $o_{\prev(Q,j)}$ is chosen
as overwrite block (and this happens with probability $1/n$). When the download phase of query $j$ starts,
block $B_{q_j}$ is found in the stash and so the probability that $d_j$ is selected as download blocks is $1/n$.
Whence
$\Pr[\ram^D_j(Q)=d_j\wedge\ram^O_{\prev(Q,j)}(Q)=o_{\prev(Q,j)}] =p/n^2.$

\smallskip\noindent{\bf Case 3.} $\prev(Q,j)=+\infty$.
In this case, $o_{\prev(Q,j)}=\perp$ and $B_{q_j}$ has not been queried before query $j$.
Therefore, at the start of the download phase of query $j$,
$B_{q_j}$ is found in block stash with probability $p$.
By using arguments similar to the ones of the previous cases,
we have
\begin{align*}
\Pr[&\ram^D_{j}(Q)=q_j\wedge\ram^O_{\prev(Q,j)}(Q)=\perp]\\
=&\Pr[\ram^D_{j}(Q)=q_j]=(1-p)+p/n
\end{align*}
and, for $d_j\ne q_j$,
\begin{align*}
\Pr[&\ram^D_{j}(Q)=d_j\wedge\ram^O_{\prev(Q,j)}(Q)=\perp]\\
=&\Pr[\ram^D_{j}(Q)=d_j]=p/n
\end{align*}
which completes the proof.
\end{proof}

\begin{proof}[Proof of Lemma~\ref{lemma:overwrite_prob}]
As in the proof of the previous lemma, we upper bound the numerator by $1$ and show that
the denominator is at least $n/p$ starting from the case $o_j=q_j$.
If, in the overwrite phase of the $j$-th query,
block $B_{q_j}$ is added to the block stash (and this happens with probability $1/n$)
then the same block is chosen as overwrite block with probability $1/n$.
In the remaining $1-p$ probability, $q_j$ is the overwrite block with probability $1$.
Therefore $\Pr[\ram^O_j(Q)=q_j]=(1-p) + p/n.$
If, instead, $o_j\ne q_i$, then the only case in which $o_j$ is the  overwrite block of phase $j$
is when $B_{q_j}$ is added to the block stash and this gives
$\Pr[\ram^O_j(Q)=o_j]=p/n,$
for every $o_j\ne q_j$.
For any $j \ne i$, $Q_j = Q'_j$ showing the first point. For any $o$,
$\Pr[\ram^O_i(Q') = o] \ge p/n$ giving the second result.
\end{proof}

%% file: generalization.tex
{
\section{$\dpram$ Generalization}
\label{sec:general}
Our proof of differential privacy of $\dpram$ of Section~\ref{sec:ram}
can be seen to give a more general form of differential privacy.
Suppose $\dpram$ stores $n$ blocks and there exists
a repertoire $\Sigma$ of size $b = O(n)$. We interpret
$\Sigma$ as the specification of $b$ buckets where each
bucket contains exactly $s$ blocks. Note, buckets may overlap and two
different buckets may contain the same block.
Our privacy proof of $\dpram$ carries over to
the case in which a query retrieves all blocks in a bucket.
In other words, we interpret query sequences $Q$ of length $s\cdot l$ as $l$ subsequences
each of length $s$ taken from the $\Sigma$ of size $n$ (that is,
each subsequence corresponds to a bucket).
In this framework, a query sequence $Q=(q_1,\ldots,q_{s \cdot l})$ is associated with
sequence $Q=(\sigma_1,\ldots,\sigma_l)$ of length $l$ over $\Sigma$; that is,
$$Q=(
    \underbrace{q_1,\ldots,q_s}_{\sigma_1},
    \underbrace{q_{s+1},\ldots,q_{2\cdot s}}_{\sigma_2},\cdots
    \underbrace{q_{(l-1)s+1},\ldots,q_{l\cdot s}}_{\sigma_l}),$$
and $\sigma_1,\sigma_2,\ldots,\sigma_l\in\Sigma$ and $|\Sigma|=b =O(n)$.

The proof of Section~\ref{sec:ram} handles the case $s=1$ in which an access sequence
is a sequence of $n$ requests taken from the repertoire $\Sigma=[n]$.

To argue the general case let us consider first the case in which the server
explicitly stores the subsequences of $\Sigma$; specifically, each element of $\Sigma$ is a sequence
of $s$ of the original $n$ blocks and each is stored using storage equal to $s$ blocks.
Then an access sequence $Q$ of length $s\cdot l$ over $[n]$ is simply an access sequence 
$(\sigma_1,\ldots,\sigma_l)$ of length $l$ over $\Sigma$ and, as it is easily seen,
our proof still works and guarantees differential privacy with $\epsilon = \Theta(\log n)$.

This approach, unfortunately, has the drawback that the server storage grows by a factor of $s$.
However, we observe that the blocks that constitute the subsequence of $\Sigma$ need not to be explicitly
stored by the server that instead can store just the $n$ original blocks
and, each time it receives a request $\sigma\in\Sigma$,
the server fetches the needed blocks. Note that this transformation preserves differential
privacy with $\epsilon = \Theta(\log n)$ as this property is
independent of the actual layout of the parts of the individual atomic blocks (in this case the subsequences
$\sigma\in\Sigma$) on which the RAM is built.
The only modification required is that when retrieving or updating
a bucket $i$, the $\dpram$ must also check if the any block
is stored on the client
as part of another bucket $j \ne i$. If so, the block on client storage must
be returned as opposed to the one from the server for retrievals.
For updates, both the server copy and client copy must be updated.
}

%% file: kvs_leftover_proofs.tex
\section{Remaining DP-KVS Proofs}
\label{sec:kvs_proofs_app}

\noindent
[Proof of Lemma~\ref{lemma:closed}]
We proceed by induction and we verify the base case by plugging in $i=0$.
For $i\geq 0$, we have
        \begin{eqnarray*}
                \beta_{i+1}&=&\frac{e}{n}\cdot \beta_{i}^2\cdot 2^{2(i+1)}\\
                           &=&\frac{e}{n}\cdot\frac{n^2}{e^2}\cdot
                                \left(\frac{2}{3}\right)^{2^{i+3}}
                                \left(\frac{1}{2}\right)^{4(i+2)}\cdot 2^{2(i+1)}\\
                           &=&\frac{n}{e}\cdot
                                \left(\frac{2}{3}\right)^{2^{i+3}}
                                \left(\frac{1}{2}\right)^{4(i+2)-2(i+1)}\\
                           &=&\frac{n}{e}\cdot
                                \left(\frac{2}{3}\right)^{2^{i+3}}
                                \left(\frac{1}{2}\right)^{2(i+3)}.
        \end{eqnarray*}

%% file: ir_up_single-app.tex
\section{DP-IR Pseudocode}
\label{sec:dpir_code}
We present the pseudocode for our $\dpir$ scheme of Section~\ref{sec:ir}.
\begin{algorithm}
	\caption{$\pirC.\!\Query$: read a block}
        \renewcommand{\algorithmicrequire}{\textbf{Input:}}
        \renewcommand{\algorithmicensure}{\textbf{Output:}}
\begin{algorithmic}
        \REQUIRE $\textIndex,\epsilon,\alpha$
	\STATE Set $\calT=\emptyset$. \qquad \COMMENT{initializing the download set}
	\STATE Pick random number $r\in[0, 1]$.
	\IF {$r > \alpha$}
		\STATE $\calT \leftarrow \calT \cup \{\textIndex\}$.
	\ENDIF
	\STATE Set $K=\lceil(1-\alpha)N/(e^\epsilon-1)\rceil$.
	\WHILE{$|\calT| < K$}
		\STATE Pick $j$ uniformly at random from $[N]\setminus \calT$.
		\STATE $\calT \leftarrow \calT \cup \{j\}$.
	\ENDWHILE
	\STATE Send $\calT$ to server and receive $\{B_j\}_{j \in \calT}$ from server.
	\IF {$r>\alpha$}
		\RETURN $B_i$.
	\ELSE
		\RETURN $\perp$.
	\ENDIF
\end{algorithmic}
\end{algorithm}

%% file: ram_up-app.tex
\section{DP-RAM Pseudocode}
\label{sec:ram_code}
We present the pseudocode for the algorithm in Section~\ref{sec:ram}.
\begin{algorithm}
        \renewcommand{\algorithmicrequire}{\textbf{Input:}}
        \renewcommand{\algorithmicensure}{\textbf{Output:}}
	\caption{$\oramC.\!\Init$: initialize client and server storage}
	\label{algo:initOram}
\begin{algorithmic}
        \REQUIRE $1^\lambda, B_1, \ldots, B_N$
	\STATE Randomly select $K\leftarrow\{0,1\}^\lambda$.
        \STATE Initialize array $A$ of size $N$ on server.
	\STATE Initialize hash table $\bStash$ on client.
	\FOR{$i \leftarrow 1,\ldots,N$}
		\STATE Set $A[i]\leftarrow\Enc(K, B_i)$.
		\STATE Pick $r$ uniformly at random from $[N]$.
		\IF {$r \le C$}
			\STATE Set $\bStash[i]\leftarrow B_i$.
		\ENDIF
	\ENDFOR
\end{algorithmic}
\end{algorithm}

\begin{algorithm}
        \renewcommand{\algorithmicrequire}{\textbf{Input:}}
        \renewcommand{\algorithmicensure}{\textbf{Output:}}
        \caption{$\oramC.\!\Query$: reading and writing a data block}
	\label{algo:queryOram}
\begin{algorithmic}
	\REQUIRE $\textIndex, \textOp, B_{\mathtt{new}}$

\COMMENT{Download Phase}

	\IF {$\textIndex \in \bStash$}
		\STATE Pick $d$ uniformly at random from $[N]$.
		\STATE Download $A[d]$ and discard.
		\STATE Set $B \leftarrow \bStash[\textIndex] $ and remove $\textIndex$ from $\bStash$.
	\ELSE
		\STATE Set $d\leftarrow i$.
		\STATE Download $A[d]$.
		\STATE Set $B \leftarrow \Dec(K, A[d])$.
	\ENDIF
\item
	\IF {$\textOp = \opWrite$}
		\STATE Set $B \leftarrow B_{\mathtt{new}}$.
	\ENDIF
\item
\item \COMMENT{Overwrite Phase}

	\STATE Draw $r$ uniformly at random from $[N]$.
	\IF {$r \le C$}
		\STATE Set $\bStash[\textIndex] \leftarrow B$.
		\STATE Pick $o$ uniformly at random from $[N]$.
		\STATE Download $A[o]$, decrypt it and re-encrypt it obtaing $\mathtt{ct}$.
		\STATE Upload $\mathtt{ct}$ as $A[o]$.
	\ELSE
		\STATE Set $o=i$.
		\STATE Download $A[o]$ and discard it.
		\STATE Set $\mathtt{ct}\leftarrow\Enc(K, B)$.
		\STATE Upload $\mathtt{ct}$ as $A[o]$.
	\ENDIF
\end{algorithmic}
\end{algorithm}